\pdfoutput=1 

\documentclass[twocolumn]{svjour3}          

\usepackage{amsfonts}
\usepackage{amsmath}
\usepackage{amssymb}
\usepackage{mathtools}
\usepackage{graphicx}
\usepackage{xcolor}
\usepackage{float}
\usepackage{stfloats}
\usepackage[toc,page]{appendix}
\usepackage{bbm}
\usepackage{bm}
\usepackage{booktabs}
\usepackage{multirow}
\usepackage{array}

\usepackage[ruled,vlined]{algorithm2e}
\SetAlFnt{\small}
\SetAlCapFnt{\small}

\usepackage{tikz}
\usetikzlibrary{shapes.geometric, arrows}
\usetikzlibrary{decorations.pathreplacing}

\usepackage{natbib}
\usepackage[hidelinks]{hyperref} 

\RequirePackage{fix-cm}
\smartqed  
\usepackage{mathptmx}      

\journalname{Statistics and Computing}

\spnewtheorem{prop}{Proposition}[section]{\bf}{\rm}
\spnewtheorem{ex}{Example}{\bf}{\rm}

\def\makeheadbox{{%
\hbox to0pt{\vbox{\baselineskip=10dd\hrule\hbox
to\hsize{\vrule\kern3pt\vbox{\kern3pt
\hbox{\bfseries pre-print}
\hbox{This is a pre-print version of this article.}
\kern3pt}\hfil\kern3pt\vrule}\hrule}%
\hss}}}

\newcommand{\x}{\textbf{x}}
\newcommand{\e}{\textbf{e}}
\newcommand{\y}{\textbf{y}}

\newcommand{\uu}{\textbf{u}}
\newcommand{\E}{\textbf{E}}
\newcommand{\X}{\textbf{X}}
\newcommand{\Y}{\textbf{Y}}

\newcommand{\ind}[1]{\boldsymbol{1}\left(#1\right)} 
\newcommand{\expect}[1]{E\left[#1\right]}             
\newcommand{\expectK}[2]{E_{#2}\left[#1\right]}       
\newcommand*{\dif}{\mathop{}\!\mathrm{d}}
\newcommand{\norm}[1]{\left\lVert #1 \right\rVert}
\DeclareMathOperator*{\argmax}{arg\,max}
\DeclareMathOperator*{\argmin}{arg\,min}

\newcommand{\eqnref}[1]{\eqref{#1}}
\newcommand{\propref}[1]{Proposition \ref{#1}}
\newcommand{\exref}[1]{Example \ref{#1}}

\newcommand{\algref}[1]{Algorithm \ref{#1}}
\newcommand{\tblref}[1]{Table \ref{#1}}
\newcommand{\figref}[1]{Figure \ref{#1}}
\newcommand{\secref}[1]{Section \ref{#1}}
\newcommand{\subsecref}[1]{subsection \ref{#1}}
\newcommand{\appref}[1]{Appendix \ref{#1}}

\allowdisplaybreaks
\begin{document}\sloppy

\title{Sequential Bayesian optimal experimental design for structural reliability analysis}

\author{Christian Agrell \and Kristina Rognlien Dahl}

\institute{
        Christian Agrell \at
            Department of Mathematics, University of Oslo, Norway\\
            \email{chrisagr@math.uio.no}\\           
            DNV GL Group Technology and Research\\
            \email{christian.agrell@dnvgl.com}           
        \and
        Kristina Rognlien Dahl \at
            Department of Mathematics, University of Oslo, Norway\\
            \email{kristrd@math.uio.no} 
}

\date{01.07.2020}

\maketitle
\begin{abstract}
    Structural reliability analysis is concerned with estimation of the probability
    of a critical event taking place, described by $P(g(\X) \leq 0)$ for some
    $n$-dimensional random variable $\X$ and some real-valued function $g$.
    In many applications the function $g$ is practically unknown, as function evaluation involves 
    time consuming numerical simulation or some other form of experiment that is expensive to perform. 
    The problem we address in this paper is how to optimally design experiments, 
    in a Bayesian decision theoretic fashion, when the goal is to estimate the probability $P(g(\X) \leq 0)$
    using a minimal amount of resources. 
    As opposed to existing methods that have been proposed for this purpose, we consider a general structural 
    reliability model given in hierarchical form. 
    We therefore introduce a general formulation of the experimental design problem, where we distinguish between 
    the uncertainty related to the random variable $\X$ and any additional epistemic uncertainty that we 
    want to reduce through experimentation. 
    The effectiveness of a design strategy is evaluated through a measure of residual uncertainty, 
    and efficient approximation of this quantity is crucial if we want to apply algorithms that search for an optimal strategy. 
    The method we propose is based on importance sampling combined with the unscented transform for epistemic 
    uncertainty propagation. We implement this for the myopic (one-step look ahead) alternative, 
    and demonstrate the effectiveness through a series of numerical experiments. 
    
\end{abstract}

\section{Introduction}
In order to ensure sufficient reliability of engineered systems, such as 
buildings, ships, offshore structures, aircraft or technological products, 
uncertainties with respect to the system's capabilities 
and the system's environment must be accounted for. 
In probabilistic structural reliability analysis, this is achieved through a 
probabilistic model of the system and its environment. 
A primary objective with such a model is to 
estimate the probability that the system will fail (e.g. collapse, sink, crash or explode).
\footnote{This is rarely interpreted as a frequentist probability. 
As the model is not the real world, it is common to design models such that 
the failure probability can be interpreted as a conservative estimate, or as 
a consistent measure of robustness for comparison with other 'acceptable' systems.}  

A probabilistic structural reliability model is commonly defined through a \emph{performance function} 
(also called a \emph{limit-state function}) $g(\X)$ depending 
on some random variable $\X$. Here, $g(\X)< 0$ corresponds to system failure, and $g(\X) \geq 0$ corresponds to the system functioning. Typically, $\X$ contains the parameters describing a particular structure, such as the geometry, dimensions and material properties. 
These quantities may be random, but can be influenced by the designer of the structure. 
For example, the designer may choose to use a more expensive, but more durable material in order to improve the structural properties of the system. 
In addition, $\X$ contains the (random) parameters that characterize the systems environment,
such as wind speed, wave height etc., and parameters describing how well the model fits reality (model uncertainties).
Given $\X$ and the function $g(\cdot)$, the \emph{probability of failure} is defined as the probability $P(g(\X) < 0)$. Modern engineering requirements for safe design and operation of such systems are usually 
given as an upper bound on this probability \citep{madsen2006methods}. 

Hence, for many practical applications, the failure probability computation is an important task.
This is often challenging for complex systems, as a computationally feasible stochastic model of the 
complete system and its environment is not available. 
In our modelling framework, this comes in the form of additional \emph{epistemic} uncertainties, 
i.e. uncertainties due to limited data or knowledge that \emph{in principle} can be reduced by gathering more information.
These epistemic uncertainties usually come in one of the two forms: 
\medskip
\begin{enumerate}
    \item The function $g(\cdot)$ or the distribution of $\X$ depends on parameters that we do not know the value of.
    \item Evaluating $g(\x)$ at some single realization $\x$ of $\X$ is \emph{expensive} in terms of money and/or time. 
\end{enumerate} 
\medskip
The last part comes from the complex physical nature of failure mechanisms, 
where experiments are needed to evaluate the function $g(\x)$. 
This includes numerical computer simulations and physical experiments in a laboratory, 
which are both time consuming and expensive. 
Hence, due to the limited number of experiments that can be performed in practice, 
any method for estimating $P(g(\X) < 0)$ that relies on a large number of 
evaluations of $g(\cdot)$ is practically infeasible.
This problem is usually solved by replacing the performance function $g(\cdot)$ with 
a computationally cheap \emph{surrogate model} or \emph{emulator}
\footnote{The word 
    \emph{emulator} is often used for a surrogate model that can interpolate between 
    noiseless observations coming from a deterministic computer simulation.
}, 
constructed from a small set of experiments. 
When the surrogate model is a stochastic process (viewed as a distribution over functions), 
we can quantify the added epistemic uncertainty that comes from this simplification. 

We will assume that epistemic uncertainty is introduced to a structural reliability model, 
and that there is a way to reduce this uncertainty by performing experiments. 
The problem we address in this paper, is how to optimally estimate $P(g(\X) < 0)$ using as little resources as possible. 
In particular, we want to find an optimal strategy for the scenario where we can perform experiments sequentially,
i.e. where each experiment may depend on the preceding ones. 
The scenario where $g(\cdot)$ is replaced by a surrogate model created from a finite set of observations 
$\{ g(\x_{i}) \}_{i=1}^n$ has already been studied extensively 
\citep{Bect:2012:Sequential_design, Echard:2011:AKMCS, Bichon:2008:EGRA, Sun:2017:LIF, Jian:2017:two_acc, Perrin:2016:AL_multiple_fm, Schueremans:2005:splines_and_NN}.
The most common approach is to approximate $g(\cdot)$ using a Gaussian process, 
and make use of 
the convenient fact that a surrogate model given by the posterior predictive distribution
of the Gaussian process has a closed form solution.
However, structural reliability models are often hierarchical, and the reason why $g(\cdot)$ is expensive 
comes from one or more expensive \emph{sub-components}
\footnote{For instance, $g(\x)$ is often a function of a structures 
    capacity and the effect of loads acting on the structure, where each of which are determined from separate types of experiments.
}. 
An example is shown in \figref{fig:ex_hierarchical_model}, where $g(\x) = g(y_{1}(\x), y_{2}(\x))$.
Assume here that $\x \in \mathbb{R}^{m}$, then the index set of the Gaussian process approximation of $g(\x)$ is $m$-dimensional. 
Naturally, the number of experiments needed is highly dependent on $m$. 
If $g(\x)$ is expensive, then this must be because one (or more) of the functions, $y_{1}(\x)$, $y_{2}(\x)$ or $g(y_{1}, y_{2})$
is expensive. Very often, the effective domains\footnote{
    If for instance $y_{1}(\x) : \mathbb{R}^{m} \rightarrow \mathbb{R}$ depends only on $x_{1}, \dots, x_{n}$ for $n \leq m$, 
    the effective domain of $y_{1}$ is $n$-dimensional.} 
of these functions have dimensionality much smaller than $m$, 
so fitting a Gaussian process to observations of $g(\x)$ is not very efficient. 
There is also some practical inconvenience here, which is that some of the expensive sub-components (for instance load models)
may be applicable in different structural reliability models, so there is a potential for re-use if 
we create a surrogate model for, say $y_{1}(\x)$, instead of $g(\x)$.
\begin{figure}[H]
    \centering
        \begin{tikzpicture}
            \tikzstyle{roundnode} = [circle, minimum size=1.0cm, text centered, draw=black]
            \tikzstyle{arrow} = [thick,->,>=stealth]

            \node (x) [roundnode] {$\x$};
            \node (g) [roundnode, below of=x, yshift=-2cm] {$g$};
            \draw [arrow] (x) -- (g);

            \node (x2) [roundnode, right of=x, xshift=2.8cm] {$\x$};
            \node (y1) [roundnode, below of=x2, yshift=-0.5cm, xshift = -1cm] {$y_{1}$};
            \node (y2) [roundnode, below of=x2, yshift=-0.5cm, xshift = 1cm] {$y_{2}$};
            \node (g2) [roundnode, below of=y1, yshift=-0.5cm, xshift = 1cm] {$g$};

            \node (lbl1) [above of=x] {Single layer model};
            \node (lbl2) [above of=x2, yshift = 0.037cm] {Hierarchical model};

            \draw [arrow] (x2) -- (y1);
            \draw [arrow] (x2) -- (y2);
            \draw [arrow] (y1) -- (g2);
            \draw [arrow] (y2) -- (g2);
        \end{tikzpicture}
    \caption{Left: Single layer model. Right: Example of an hierarchical (2 layer) model where $g(\x) = g(y_{1}(\x), y_{2}(\x))$.}
    \label{fig:ex_hierarchical_model}
\end{figure}
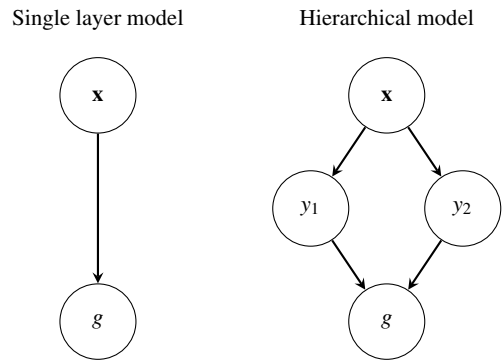
In this paper we will work with hierarchical models as the one illustrated in \figref{fig:ex_hierarchical_model}, 
where we assume that some of the intermediate variables are stochastic processes with epistemic (potentially reducible) 
uncertainty. Note that this also covers case where we just introduce additional epistemic variables into the model. 
Actually, in the approximate numerical solution we propose in this paper, these two problems become equivalent. 
Moreover, as Gaussianity generally is lost in the hierarchical setting, we will only make assumptions 
on existence of second order moments of the stochastic processes used as surrogates.
We will present a general formulation of the problem of finding an optimal strategy for performing 
experiments based on Bellman's principle of optimality,
and discuss some alternative routes for solving such problems. 
For the myopic (one step look-ahead) strategy, we propose an efficient numerical procedure, 
based on finite dimensional approximation of the stochastic processes and uncertainty 
propagation using the unscented transform.  

The structure of the remaining part of the paper is as follows: 
Through \secref{sec:problem_formulation} and \secref{sec:modelling_information} we develop the Bayesian optimal experimental design problem 
for a general structural reliability model. We introduce a framework for separation of aleatory and 
epistemic uncertainties using conditional expectations, from which we can express any type of experiment 
associated with a structural reliability problem.
For the purpose of estimating a failure probability, we consider three alternative optimization objectives, 
and in \secref{sec:modelling_information} we 
discuss how the experimental design problem may be tackled using dynamic programming 
and the myopic approximation. Optimization problems of this form will involve evaluation of a
\emph{measure of residual uncertainty}, and in \secref{sec:approx_H} we present an approach 
for approximating this quantity. 
We implement this in \secref{sec:myopic_num} to develop an efficient numerical procedure 
for myopic scenario, which we illustrate through a series of examples in \secref{sec:num_exp}. 
Finally, our concluding remarks are given in \secref{sec:concluding_remarks}, and 
some supporting material used throughout the paper is included in the Appendices. 

\section{Problem formulation}
\label{sec:problem_formulation}
Given a probabilistic surrogate of a structural reliability model, we are interested in 
how to optimally improve the model for failure probability estimation, given a fixed experimental budget. 
More generally, given a structural reliability model with epistemic uncertainty (e.g. as introduced when using a surrogate),
and a set of possible experiments than can be performed, we want to select the experiments in an optimal manner. 
The choice of experiment is called a \emph{decision}, $d \in \mathbb{D}$ where $\mathbb{D}$ is a space of feasible decisions. Note that this set may include different kinds of decisions, such as performing computer experiments, lab experiments or performing physical measurements in the field. 

In the following subsections we present a rigorous formulation of the Bayesian optimal experimental design problem for structural reliability analysis.
Here we will need a way to express uncertainty about the \emph{performance function} used in structural reliability models, 
and a way to model uncertainty about future outcomes of potential experiments that can be made. 
For this purpose we will define a model $(\xi, \delta)$, where 
\begin{itemize}
    \item $\xi$ is a stochastic representation of the \emph{performance function} $g(\x)$
    evaluated at some fixed input $\x$. 
    \item $\delta(d)$ is a predictive model of experimental outcomes given a decision $d$.  
    In other words, $\delta$ models the data generating process of potential experiments.
\end{itemize}
We will consistently write $\X$ as a random variable with values in $\mathbb{R}^{m}$, 
and let $\x$ be a deterministic realization. 
$\xi$ and $\delta$ are stochastic processes, indexed over inputs $\x$ and decisions $d$ respectively. 
In structural reliability analysis, we are interested in the random variable $g(\X)$,
and likewise we will consider $\xi(\X)$, but now where $\xi(\x)$ is also random for any fixed $\x$. 
As the purpose of performing experiments will be to provide information about $\xi$, note that
$\xi$ and $\delta$ are generally not independent.

A detailed description of how $(\xi, \delta)$ is constructed is provided in the following subsections. 

\subsection{Structural reliability analysis}
\label{sec:SRA}
Let $\mathbb{X} \subseteq \mathbb{R}^m$, and let $\X$ be a random variable on the probability space $(\Omega, \mathcal{F}, P)$ with values in $\mathbb{X}$
and $g : \mathbb{X} \rightarrow \mathbb{R}$ a measurable function. 
We call $g$ the \emph{performance function} or \emph{limit state}, with the associated failure set 
\begin{equation*}
    F_{g} = \{ \x \in \mathbb{X} \ | \ g(\x) \leq 0 \}.
\end{equation*}
In structural reliability analysis, we are interested in estimating the \emph{failure probability},
which we here denote $\bar{\alpha}$. It is defined as 
\begin{equation}
    \label{eq:pof_def}
    \bar{\alpha}(g) = P(F_{g}) = \expect{\ind{g(\X) \leq 0}},
\end{equation}
where $\expect{\cdot}$ denotes the expectation with respect to $P$ and $\ind{\cdot}$ is the indicator function. 

In most real-world cases it is difficult to derive an analytical expression for the failure probability. To overcome this, several approximation and simulation methods have been suggested, 
see e.g. \cite{madsen2006methods} or \cite{Huang:2017:SRA}. 
Two traditional methods are the first- and second-order reliability method (FORM/SORM), where 
the failure boundary is approximated at a specific point using a Taylor expansion up to the first and second order 
respectively. Different sampling procedures have also been developed, which often make use of intermediate results 
obtained from FORM/SORM. 
Other relevant techniques involve the construction of environmental contours and
the estimation of buffered failure probabilities as in \citep{DahlHuseby}.
In this paper, our focus is different from these methods in the sense that we are mainly interested in how to estimate the failure probability as well as possible, given a limited experimental budget. To do so, we need to separate between different kinds of uncertainty in our model.

\subsection{Separating epistemic and aleatory uncertainties}
\label{sec:al_ep}
Ideally, the uncertainty related to the random variable $g(\X)$ in \eqnref{eq:pof_def} is \emph{aleatory}, 
in the sense that that it relates to inherent variability of the physical 
phenomenon that is being modelled, but in reality we must also include \emph{epistemic}
uncertainty due to lack of information or knowledge. 
For instance, assume that $g(\x, \e)$ depends on the aleatory variable $\x$ and some fixed but unknown parameter $\e$. 
Assume further that $\X$ is the aleatory random variable representing variability in $\x$, 
$\E$ is the epistemic random variable representing our belief about $\e$, and that 
$\X$ and $\E$ are independent with laws $P_{x}$ and $P_{e}$. 
It is then relevant to view the failure probability as a random quantity with epistemic uncertainty, 
$\alpha(\E) = \int \ind{g(\x, \E) \leq 0} P_{x}(d\x)$. 
For engineering applications, one would then typically be interested in some specified upper 
percentile values of $\alpha(\E)$, i.e. ensuring that the epistemic uncertainty is under control.

In the following, we will assume that we have a performance function $\xi(\cdot)$ that depends on 
a strictly aleatory random variable $\X$, and some other random quantity with epistemic uncertainty.
We will need to formulate this with a bit of generality, in order to cover the different ways 
epistemic uncertainty can be introduced in a structural reliability model. 

As in \secref{sec:SRA} we will work with $(\Omega, \mathcal{F}, P)$ as the global probability space, 
capturing all forms of uncertainty. We then let $\mathcal{A}$ and $\mathcal{E}$ be two 
sub $\sigma$-algebras representing respectively aleatory and epistemic information, 
and we assume that $\X$ is $\mathcal{A}$-measurable. 
Furthermore, for any $\X \in \mathbb{X}$ we assume that $\xi(\X)$ is $\mathcal{E}$-measurable. 
That is, $\xi : \mathbb{X}\times\Omega \rightarrow \mathbb{R}$ is a stochastic process indexed by $\X \in \mathbb{X}$ (this is also called a random field), 
and $\xi(\X)$ is a real-valued random variable. 
We will write $\xi(\cdot)$ instead of $g(\cdot)$ whenever epistemic uncertainty has been introduced, 
as for instance in the canonical case where a deterministic performance function $g(\cdot)$ is approximated with a probabilistic surrogate $\xi(\cdot)$.

We can now define the failure probability with epistemic uncertainty as the $\mathcal{E}$-measurable random variable 
\begin{equation}
    \label{eq:pof_def_cond}
    \alpha(\xi) = \expect{\ind{\xi(\X) \leq 0} \ | \ \mathcal{E}}.
\end{equation}
Note that \eqnref{eq:pof_def_cond} coincides with \eqnref{eq:pof_def} in the case where the performance 
function is not affected by epistemic uncertainty, and in general as $\bar{\alpha}(\xi) = \expect{\alpha(\xi)}$ because
\begin{equation}
    \begin{array}{lll}
    \expect{\alpha(\xi)} &=& \expect{\expect{\ind{\xi(\X) \leq 0}} \ | \ \mathcal{E}} \\
    &=& \expect{\ind{\xi(\X) \leq 0}} \\
    &=& \bar{\alpha}(\xi),
    \end{array}
\end{equation}
\noindent where the second equality uses the double expectation property.

In the following we will just write $\alpha$ or $\bar{\alpha}$ without the dependency on $\xi$ when there is no risk of confusion. 
\begin{ex}
    Assume $\xi$ is a deterministic function of the aleatory random variable $\X$ and epistemic random variable $\E$, 
    both defined on $(\Omega, \mathcal{F}, P)$.
    Then $\mathcal{A} = \sigma(\X)$ and $\mathcal{E} = \sigma(\E)$, i.e., the $\sigma$-algebras generated by the random variables $\X$ and $\E$ respectively.
    \label{ex:two_rv}
\end{ex}
Note that the converse of \exref{ex:two_rv} also holds true, as we can always view 
$\xi$ as a deterministic function applied to two random variables $\X$ and $\E$. 
That is, where $\xi(\x, \e)$ is a deterministic function for $\x$ and $\e$ fixed, 
and we can write the stochastic process $\xi(\x, \omega)$ as $\xi(\x, \E)$. It is sometimes useful 
to think of $\xi$ in this way. In particular, the numerical approximation we propose later 
in this paper is based on obtaining a finite dimensional approximation of $\E$.
\begin{ex}
    Let $g$ be given as in the hierarchical model in \figref{fig:ex_hierarchical_model}, and 
    $\X$ a random variable defined on some measure space $(\Omega_{x}, \mathcal{F}_{x}, P_{x})$. Assume that $y_{1}$ and $y_{2}$ are 
    expensive to evaluate, so we replace them with surrogate models in the form of two stochastic processes 
    $\widetilde{y_{1}}$ and $\widetilde{y_{2}}$ defined on another measure space $(\Omega_{y}, \mathcal{F}_{y}, P_{y})$. Note that we assume that both $\widetilde{y_{1}}$ and $\widetilde{y_{2}}$ are defined on the same measure space.
    Then, the measure space for the experimental design problem is given by $(\Omega, \mathcal{F}, P) = (\Omega_{x} \times \Omega_{y}, \mathcal{F}_{x} \otimes \mathcal{F}_{y}, P_{x} \times P_{y})$, 
    $\mathcal{A} = \mathcal{F}_{x}$ and $\mathcal{E} = \mathcal{F}_{y}$ (up to isomorphism), 
    and we would write $\xi(\x) = g(\widetilde{y_{1}}(\x), \widetilde{y_{2}}(\x))$.
    \label{ex:stoch_proc}
\end{ex}

\subsection{Decisions, outcomes and experiments}
We are interested in the case where the epistemic uncertainty in $\alpha$ can be reduced
by running experiments. For instance, in \exref{ex:two_rv} the epistemic variable $\E$ could 
be a fixed but unknown parameter, and maybe additional measurements could be performed 
to reduce the uncertainty in $\E$. Or in \exref{ex:stoch_proc}, additional experiments could 
be performed to infer the values of $y_{1}$ or $y_{2}$ at some given input $\x'$, in 
order to reduce uncertainty in the surrogate models $\widetilde{y_{1}}$ and $\widetilde{y_{2}}$.

These are examples of possible \emph{decisions} we could make to reduce epistemic uncertainty. 
We will let $\mathbb{D}$ denote the set of all possible decisions, and $\mathbb{O}$ the 
set of all possible \emph{outcomes}. For any decision $d \in \mathbb{D}$, the corresponding 
outcome is uncertain a priori, and in order to evaluate the potential impact of a decision 
we will need to specify (possibly subjectively) a distribution representing the possible outcomes.
We will let $\delta(d)$ denote the random outcome of a decision $d \in \mathbb{D}$ with values in $\mathbb{O}$. 
For any realization $o \in \mathbb{O}$ of $\delta(d)$, we will refer to the pair $(d, o)$ as an \emph{experiment}. 

In our modelling framework, we will assume that $\xi(\x)$ as defined in \secref{sec:al_ep} is provided 
together with $(\Omega, \mathcal{F}, P)$ and the sub $\sigma$-algebras $\mathcal{A}$ and $\mathcal{E}$, 
and that a decision process $\delta(d)$ is given where $\delta(d)$ is $\mathcal{E}$-measurable for any $d \in \mathbb{D}$. 
Table \ref{table: notation} gives an overview of the notation we have introduced so far, in order to define the problem of optimal experimental design for structural reliability analysis.
\begin{ex}
    Continuing from \exref{ex:stoch_proc}, assume that noise perturbed observations of $y_{1}$ 
    can be made. Let $d(\x) = \{ \text{observe } y_{1}(\x) \}$, and define $\mathbb{D}$ as the union of such events for all $\x$. 
    If we assume that observations come with additive noise, $o(\x) = y_{1}(\x) + \epsilon(\x)$, for some 
    specified noise process $\epsilon$, then we can let $\delta(d(\x)) = \widetilde{y_{1}}(\x) + \epsilon(\x)$.
    In a similar fashion, $\mathbb{D}$ and $\delta(d)$ could be extended to include observations of $y_{2}$ as well.
    \label{ex:obs_proc}
\end{ex}
We will note that the noise-free alternative to \exref{ex:obs_proc}, i.e. the case where $\epsilon \equiv 0$, 
is a common scenario when dealing with deterministic computer simulations. 
Another related scenario that is also of relevance here, is that of muiltifidelity modelling \citep{Fernandez:2017:MFM_review}, 
in which case inaccurate estimates of $y_{1}(\x)$ could be available at the same time, but at a lower cost. 
\begin{table*}[ht]
    \footnotesize
    \centering
    \caption{Overview of the framework for the optimal experimental design problem for structural reliability analysis} 
    \begin{tabular*}{\textwidth}{ >{\centering\arraybackslash}p{0.1\textwidth} | p{0.5\textwidth} | p{0.4\textwidth} }
    \toprule
    Symbol & Description & Type \\
      \hline
      $\X$ & Parameters describing structure and environment & $\mathbb{R}^m$-valued random variable  \\
      $\x$ & Deterministic realization of $\X$ & values in $\mathbb{R}^m$ \\
      $g(\x)$ & Performance function of structure & real-valued \\
      $F_g$ & Failure set of $g(\cdot)$ & subset of $\mathbb{R}^m$ \\
      $\bar{\alpha}(g)$ & The failure probability, $P(F_g)$ & values in $[0,1]$ \\
      $\xi(\x)$ & Stochastic approximation of $g(\cdot)$ & real-valued stochastic process \\
      $\alpha(\xi)$ & The failure probability with epistemic uncertainty & values in $[0,1]$ \\
      \hline
      $d$ & Decision & contained in set of decisions $\mathbb{D}$ \\
      $o$ & Outcome of experiment & contained in set of outcomes $\mathbb{O}$ \\
      $(d,o)$ & Summary of an experiment & contained in $\mathbb{D} \times \mathbb{O}$\\
      $\delta(d)$ & Model of experiment outcomes & $\mathbb{O}$-valued stochastic process  \\
      \hline
      \E & Parameters for epistemic uncertainty, independent of $\X$ & random variable  \\
      $\mathcal{A}$ & Aleatory information & $\sigma$-algebra \\
      $\mathcal{E}$ & Epistemic information & $\sigma$-algebra \\
      $(\xi, \delta)$ & The model &  $(\mathbb{R}\times\mathbb{O})$-valued \\
      \bottomrule
    \end{tabular*}
    \label{table: notation}
\end{table*}

\subsection{Sequential model updating}
\label{sec:model_update}
Now, having defined a random variable $\X$ and the two processes $\{\xi(\x)\}_{\x \in \mathbb{X}}$ and $\{\delta(d)\}_{d \in \mathbb{D}}$, 
we want to perform a sequence of experiments, $(d_{0}, o_{0}), (d_{1}, o_{1}), \dots$, 
and update $\xi$ and $\delta$ accordingly. 

We let $I_{k} := \{ (d_{0}, o_{0}), \dots, (d_{k-1}, o_{k-1}) \}$
denote the information or history up to the $k$-th experiment, and define $\mathcal{E}_{k}$ as the $\sigma$-algebra 
generated by $\mathcal{E}$ and $I_{k}$. Hence, $\mathcal{E}_{k}$ is all the information regarding epistemic quantities 
that is available after $k$ experiments. We introduce the notation $P_{k}(\cdot)$ and $\expectK{\cdot}{k}$ to 
denote the conditional distribution $P(\cdot \ | \ \mathcal{E}_{k})$ and conditional expectation $\expect{\cdot \ | \ \mathcal{E}_{k}}$
given the updated information $\mathcal{E}_{k}$.
For convenience we define $I_{0} = \emptyset$, so that we can use the index $k=0$ with these definitions 
for the scenario before any experiment has been made. 
We will write $\xi_{k}$ and $\delta_{k}$ as the updated processes $\xi | I_{k}$ and $\delta | I_{k}$ 
corresponding to $P_{k}$. Per definition, 
\begin{equation*}
    (\xi_{k+1}, \delta_{k+1}) =  (\xi_{k}, \delta_{k}) \ | \ d_k, o_k
    = (\xi_{0}, \delta_{0}) \ | \ I_k, d_k, o_k.
\end{equation*}
In the following example, we show how this sequential update can be done via Bayes' theorem.
\begin{ex}
    Let $k \in \mathbb{N}$. Assume $(\xi, \delta)$ admits a joint probability density at any finite subset of $\mathbb{X}\times \mathbb{D}$  
    with respect to $P_{k}$, which we write $p_{k}(\xi, \delta)$ for short. 
    E.g. $p_{k}(\xi)$ means $$P_{k}\left( \left( \xi(\x^{(1)}), \dots, \xi(\x^{(n)}) \right) = \left( \xi^{(1)}, \dots, \xi^{(n)} \right) \right)$$
    for some $\x^{(1)}, \dots, \x^{(n)} \in \mathbb{X}$ and $\xi^{(1)}, \dots, \xi^{(n)} \in \mathbb{R}$. Then $p_{k}(\xi) = p_{0}(\xi_{k})$, $p_{k}(\delta) = p_{0}(\delta_{k})$, and the update of the probabilities is done by using Bayes' theorem:
    \begin{equation}
    \begin{split}
            p_{k+1}(\xi)  = p_{k}(\xi | d_k, o_k) = \frac{p_{k}(o_k | \xi, d_k)p_{k}(\xi)}{p_{k}(o_k | d_k)}, \\ 
            p_{k+1}(\delta) = p_{k}(\delta | d_k, o_k) = \frac{p_{k}(o_k | d_k, \delta) p_k(\delta)}{p_{k}(o_k | d_k)},
    \end{split}
    \end{equation}
    where $p_{k}(\cdot | \cdot)$ is the relevant density with respect to $P_{k}$.
\end{ex}
\begin{ex}
    For a specific problem there will typically be simpler ways of updating the model than the generic 
    formulation given in the previous example. 
    Continuing again from \exref{ex:stoch_proc} and \exref{ex:obs_proc}, assume $\delta(d) = \delta(\x, \widetilde{y_{1}}, \widetilde{y_{2}})$ corresponds to 
    observing $\widetilde{y_{1}}(\x) + \epsilon_{1}(\x)$ or $\widetilde{y_{2}}(\x) + \epsilon_{2}(\x)$. 
    Then $\widetilde{y_{1}}$ and $\widetilde{y_{2}}$ can be updated directly, and we 
    let $\xi | I_k = g(\widetilde{y_{1}} | I_k, \widetilde{y_{2}} | I_k)$ and $\delta | I_k = \delta(\x, \widetilde{y_{1}} | I_k, \widetilde{y_{2}} | I_k)$.
    
    In fact, if $\widetilde{y_{1}}$ and $\widetilde{y_{2}}$ and the noise terms $\epsilon_{1}$ and $\epsilon_{2}$ 
    are all Gaussian processes, then $\widetilde{y_{1}} | I_k$ and $\widetilde{y_{2}} | I_k$ are also Gaussian
    and closed form representations are available (see \appref{app:GP}).
    Note that in this case the model update could include updating the Gaussian process hyperparameters as well. 
\end{ex}

\subsection{Optimization objective}
\label{sec:opt_obj}
Following the formulation of \cite{Bect:2012:Sequential_design, Bect:2019:Supermartingale}, 
a strategy for uncertainty reduction starts with a \emph{measure of residual uncertainty} for the quantity of interest 
after $k$ experiments. This is a functional 
\begin{equation}
    H_{k} = \mathcal{H}(P_k)
\end{equation}
of the conditional distribution $P_k$. In this paper we will consider three specific alternatives for $H_k$. 

Assume $k$ experiments have been performed, resulting in the updated probabilistic model $(\xi_k, \delta_k)$.
The updated failure probability according to \eqnref{eq:pof_def_cond} can then be defined as
\begin{equation}
    \alpha_k = \alpha(\xi_k) = \expectK{\ind{\xi(\X) \leq 0}}{k}, \ \ \bar{\alpha}_k = \expect{\alpha_k}.
\end{equation}
As we are interested in reducing uncertainty in $\alpha$, a natural optimization objective is to minimize
$\text{Var}(\alpha_k) = \expect{(\alpha_k - \bar{\alpha}_k)^2}$. 
However, computation of $\text{Var}(\alpha_k)$ can be problematic in practice. Most of the proposed 
methods for design of experiments in (non-hierarchical) structural reliability models 
therefore make use of alternative \emph{heuristic} optimization objectives. That is, some 
alternative function $H_k(\cdot)$ that is easier to compute than $\text{Var}(\alpha_k)$, 
and where the design that minimizes $H_k(\cdot)$ hopefully also performs well with respect to $\text{Var}(\alpha_k)$.

\cite{Bect:2012:Sequential_design} present a few such criteria, some of which
will also be considered in this paper. Let 
\begin{equation}
    \label{eq:pk_lambda_k}
    \begin{split}
        &p_k(\X) = P_k(\xi(\X) \leq 0), \\ &\gamma_k(\X) = p_k(\X)(1 - p_k(\X)).
    \end{split}
\end{equation}
Observe that 
\begin{equation}
\begin{array}{lll}
\text{Var}(\ind{\xi_k(\x) \leq 0}) &=& \expect{(\ind{\xi_k(\x) \leq 0})^2} - \\
& & \expect{\ind{\xi_k(\x) \leq 0}}^2\\
&=& \expect{(\ind{\xi_k(\x) \leq 0})} - p_k(\x)^2 \\
&=& p_k(\x) - p_k(\x)^2 \\
&=& \gamma_k(\x),
\end{array}
\end{equation}
\noindent and also that $\gamma_k(\x) / 2$ is the probability that two i.i.d. samples from $\xi_k(\x)$ have the same sign. 
Hence, $\gamma_k$ provides a measure of how accurate $\xi_k(\x)$ is around the critical value $\xi_k = 0$.
We will introduce two measures of residual uncertainty based on taking the expectation of $\gamma_k$ with respect 
the distribution of $\X$, which we denote $P_{\mathbb{X}}$. 
In total, we will consider the following three alternatives for $H_k$:
\begin{equation}
    \begin{split}
        H_{1, k} &= \expectK{(\alpha - \bar{\alpha})^2}{k}, \\ 
        H_{2, k} &= \int_{\mathbb{X}} \gamma_k \dif P_{\mathbb{X}} = \expect{\gamma_k}, \\
        H_{3, k} &= \left(\int_{\mathbb{X}} \sqrt{\gamma_k} \dif P_{\mathbb{X}}\right)^2 = \expect{\sqrt{\gamma_k}}^2.
    \end{split}
\end{equation}
Here $H_{2, k}$ and $H_{3, k}$ can also be motivated by realizing that 
they serve as upper bounds on $H_{1, k}$. In fact, $H_{1, k} \leq H_{3, k} \leq H_{2, k}$
(see Proposition 3 in \citep{Bect:2012:Sequential_design}).

For optimal design of experiments we will consider loss functions given by the above measures of 
residual uncertainty, potentially in combination with an additional penalty term that represents the 
cost of performing a given experiment. 
In the Bayesian decision-theoretic framework, given such a loss function depending on a \emph{policy} for selecting experiments $\pi$, 
we can evaluate the policy by looking $n$-steps ahead. 
For instance, a relevant loss function for minimizing uncertainty in $\alpha$ after $n$ additional 
experiments, following after the current experiment $k$, could be given as $J_k(\pi) = \expectK{H_{1, k+n}}{k}$ 
where $\mathcal{E}_{k+n}$ corresponds to following the policy $\pi$. The additional notation introduced
with respect to the measure of residual uncertainty and sequential model updating is summarized in \tblref{table: notation_sequential}.

\begin{table*}[ht]
    \footnotesize
    \centering
    \caption{Overview of the framework for the optimal experimental design problem for 
    structural reliability analysis with sequential model updating}
    \begin{tabular*}{\textwidth}{ >{\centering\arraybackslash}p{0.15\textwidth} | p{0.35\textwidth} | p{0.5\textwidth} }
    \toprule
    Symbol & Description & Type \\
      \hline
      $I_k$ & Information up to $k$'th experiment & sequence of decisions and outcomes\\ 
      $\mathcal{E}_k$ & Information given $\xi$ and $I_k$ & $\sigma$-algebra \\
      $P_k$ & Conditional probability given $\xi_k$ & values in $[0,1]$ \\
      $\xi_k$ & Update of $\xi$ given $I_k$ & stochastic process indexed by $k$ \\
      $\delta_k$ & Update of $\delta$ given $I_k$ & stochastic process indexed by $k$ \\
      $H_k = \mathcal{H}(P_k)$ & Measure of residual uncertainty & functional from space of probability distributions to $\mathbb{R}$\\
      $\alpha_k$ & Updated epistemic failure probability & values in $[0,1]$ \\
      $\bar{\alpha}_k$ & Updated expected failure probability & values in $[0,1]$ \\
      \bottomrule 
    \end{tabular*}
    \label{table: notation_sequential}
\end{table*}

\section{Modelling information and experimental design}
\label{sec:modelling_information}
In this section, we introduce the experimental design framework and explain how the development of information is modelled in this context. 
In the following, let $k=0, 1, \ldots, K-1$ be the \emph{experiment index} which keeps track of the number of performed experiments.

\subsection{The dynamic programming formulation}
\label{sec:dynamic}
\cite{Huan:2016:SBOED_dyn_prog} introduce a general framework for sequential optimal experimental design: 
Let the \emph{state}\footnote{
    In \citep{Huan:2016:SBOED_dyn_prog} the \emph{state} is written as $s_k = (s_k^{(b)}, s_k^{(p)})$, 
    where $s_k^{(b)}$ denotes the \emph{uncertainty state} and  
    $s_k^{(p)}$ denotes the \emph{physical state} that describes any additional deterministic decision-relevant variables.
    Herein we will not write $s_k$ specifically in this form.
} 
\emph{of the system} after experiment $k-1$ be denoted by $s_k$. 
The input (decided by the experimental designers) to experiment $k$ is denoted by $d_k$. 
We want to determine a \emph{policy} 
\begin{equation*}
    \pi := (\pi_0, \pi_1, \ldots, \pi_{K-1})    
\end{equation*}
where $d_k=\pi_k(s_k)$. That is, given the current state of the system, 
the policy is a function which tells the experimental designer the input to the next experiment.

From each experiment, we get \emph{observations} $o_k$. These observations may include measurement 
noise and modelling errors. Associated to each experiment, we have a \emph{stage reward} $R_k(s_k, o_k, d_k)$. 
The stage reward reflects the cost of doing the experiment (measured in e.g. money or time) plus any additional 
benefits or penalties of doing the experiment (measured in the same unit). 
Furthermore, we have a \emph{terminal reward} $R_K(s_K)$ only depending on the final state of the system. 

In order to model the development of the system of experiments, we have the \emph{system dynamics}:
\begin{equation*}
    s_{k+1} = \mathcal{F}(s_k, d_k, o_k)    
\end{equation*}
where $\mathcal{F}(\cdot)$ is some function specifying the transition from a current state to a new state based on the performed experiment.

The optimal experimental design problem can then be formulated as follows:
\begin{equation}
\label{eq:dyn_prog}
    \begin{split}
        &\textit{Maximize} \\ 
        &E_{o_0, \ldots, o_{K-1}} \left[ \sum_{k=0}^{K-1}  R_k(s_k,o_k,\pi_k(s_k)) + R_K(s_K) \right] \\
        &\textit{such that} \\ 
        &s_{k+1} = \mathcal{F}(s_k, d_k, o_k),
    \end{split}
\end{equation}

and the maximization is done over all policies $\pi$ that do not look into the future. 
That is, when deciding policy $\pi_k$, only what is know up to experiment $k-1$ can be used. 
Another way of saying this is that the policy $\pi$ should be \emph{adapted} to the filtration 
generated by the processes $\{s_k\}, \{o_k\}$ and $\{d_k\}$. Note that the optimization problem is over the whole 
experimental design period and is based on that there is an initial number $K$ of experiments that are to be performed. 
The experimental design policy is chosen in order to maximize the total reward of all of the experiments, 
as opposed to simply doing what is optimal in the next experiment (without taking the future experiments into account). 

To adapt this framework to the experimental design problem for structural reliability analysis, we write
\begin{equation}
    \begin{array}{lll}
        s_k = (\xi_k, \delta_k, I_k), \\[\medskipamount]
        d_k = \pi_k(s_k), \\[\medskipamount]
        o_k = \delta_k (d_k),
    \end{array}
\end{equation}
and where the dynamics $s_{k+1} = \mathcal{F}(s_k, d_k, o_k)$ is given by updating $\xi_k$, $\delta_k$ and $I_k$
with respect to the experiment $(d_k, o_k)$ as described in \subsecref{sec:model_update}.

\begin{remark}
\label{remark: backwards}
    Note that the expectation in \eqref{eq:dyn_prog} is with respect to future outcomes $o_0, \ldots, o_{K-1}$
    which a priori are uncertain, and where each outcome $o_k$ depends on the previous outcomes $o_0, \dots, o_{k-1}$.
    An equivalent formulation can be given in terms of conditional expectations.  
    Let each reward be defined by backwards induction:
    \begin{equation*}
        R_k = \max_{d} \expectK{R_{k+1} | d_{k}=d}{k}, \ \  k = K-1, \dots, 0,
    \end{equation*}
    where $R_{K} = R_{K}(s_k)$ only depends on the final state of the system. 
    Then, the policy defined by selecting for each $k$ the decision
    \begin{equation*}
        \begin{split}
        d^{*}_{k} 
        &= \argmax_{d \in \mathbb{D}} \expectK{R_{k+1} | d_k = d }{k} \\
        &= \argmax_{d \in \mathbb{D}} \expectK{\max E_{k+1} \cdots E_{K}R_K | d_k = d }{k}    
        \end{split}
    \end{equation*}
    is optimal. This corresponds with the formulation used by \cite{Bect:2012:Sequential_design}.
\end{remark}

Problem \eqref{eq:dyn_prog} is a dynamic programming problem. 
Though theoretically optimal, such problems are known for suffering form the so-called \emph{curse of dimensionality}. 
That is, the number of possible sequences of design and observation realizations grow exponentially with the dimension of the state space. 
According to \cite{Defourny:2011:Multistage_stoch_prog}, the curse of dimensionality implies that dynamic programming can only be 
solved numerically for state spaces embedded in $\mathbb{R}^d$ with $d \leq 10$. 
Therefore, such problems can often only be solved approximately via \emph{approximate dynamic programming}, see \cite{Huan:2016:SBOED_dyn_prog}. Note also that this type of formulation is based on a \emph{Markovianity} assumption, i.e., that there is no memory in the dynamics of the system. This assumption is necessary in order to perform the simplification to only having dependency on the current state of the system in Remark \ref{remark: backwards}. If the system is not Markovian, in the sense that the decision at any time depends not only on the current state of the system, but also on some of the previous states of the system, we cannot solve the experimental design problem by backwards induction. The reason for this is that the Bellman equation, which backwards induction is based on, does not hold in this case. In such cases, the experimental design problem can for instance be solved via the maximum principle, see e.g. \cite{Dahletal} for an example of systems with memory in continuous time.

\begin{remark}
    An alternative solution method to dynamic programming for problem \eqref{eq:dyn_prog} 
    is to use a scenario tree based approach, see \cite{Defourny:2011:Multistage_stoch_prog}. 
    Scenario tree based approaches are not sensitive to curse of dimensionality based on the state space, 
    but based on the number of experiments. Hence, a scenario based approach can be attempted whenever there are few experiments 
    (less than or equal $10$), but potentially a large dimensional state space (greater than $10$). 
    If the number of experiments is large (greater than $10$), but the state space dimension is small (less than or equal $10$), 
    dynamic programming is a viable solution method. If both the state space dimension and the number of experiments is large, 
    one can try approximate dynamic programming (see \cite{Huan:2016:SBOED_dyn_prog}) or a myopic formulation as an 
    alternative to the dynamic programming one. In Section \ref{sec:myopic}, we consider such a myopic formulation.
\end{remark}

Note that problem \eqref{eq:dyn_prog} is maximization problem of a reward, but can trivially 
be transformed to a minimization problem with some loss function $L_{k} = -R_{k}$ instead. 
For the application considered in this paper, we are interested in minimization problems
associated with the residual uncertainty described in \secref{sec:opt_obj}. 

\begin{ex}
    Let $\lambda(d_k)$ denote the cost of decision $d_k$. A relevant set of loss functions 
    could then be: $L_{k}(s_k, d_k,$ $o_k)$ $= 0$ for $k < K$ and $L_K = H_K \cdot \sum_{k < K} \lambda(d_k)$,
    where $H_K = H_{1, k}, H_{2, k}$ or $H_{3, k}$ as described in \secref{sec:opt_obj}. 
    Or, letting $L_{k}(s_k, d_k, o_k) = \eta^k \lambda(d_k)H_{k}$ for $k < K$ where $\eta$
    is some discount factor, $\eta \in (0, 1)$, would produce a similar but more greedy policy.
    Another relevant alternative is to define $L_K = \sum_{k < k^*} \lambda(d_k)$ as the sum of costs 
    up to the iteration $k^*$ where some target level, $H_{k} < H^*$ for $k > k^*$, has been reached.  
\end{ex}

\subsection{The myopic formulation}
\label{sec:myopic}

As mentioned in Section \ref{sec:dynamic}, the dynamic programming formulation of the optimal experimental design problem for structural reliability analysis suffers from the curse of dimensionality. An approximation to the dynamic programming formulation which mends this problem, is the \emph{myopic formulation}. This corresponds to truncating the dynamic programming sum in \eqref{eq:dyn_prog} and only looking at one time-step ahead at the time. Due to the truncation, the myopic formulation is not theoretically optimal, but it is computationally feasible even for large systems since it does not suffer from the curse of dimensionality.

In this section, we define the 
the \emph{myopic optimal decision} $d \in \mathbb{D}$ at step $k$ as the minimizer 
of the following function 
\begin{equation}
    \label{eq:acquisition}
    J_{i, k}(d) = \lambda(d)\expectK{H_{i, k+1}}{k, d} \textit{ for } i = 1, 2, \textit{or } 3.
\end{equation}
Here $H_{i, k}$ are the measures of residual uncertainty defined in \secref{sec:opt_obj}, and 
$E_{k, d}$ represents the conditional expectation with respect to $\mathcal{E}_{k}$ with $d_k = d$.
Hence, $\expectK{H_{i, k+1}}{k, d}$ represents how desirable decision $d$ is for reducing the
expected remaining uncertainty in $\alpha$ at experiment $k+1$, if the next experiment is performed with input $d$.
We let $\lambda(d)$ be a deterministic function representing the cost associated with decision $d$, 
and we will refer to a function $J_{i, k}(d)$ as the \emph{acquisition} function for myopic design. 
Other ways of introducing additional rewards or penalties associated with an experiment are of course also possible.
In fact, there is no particular reason why we write \eqnref{eq:acquisition} as a \emph{product} of cost 
and the measure of residual uncertainty, besides emphasizing that $J_{i, k}(d)$ should be
a function of these two terms. 

Note that this is greedy strategy. 
At each time step, we choose the input for the next experiment which is optimal given that we 
only look one step ahead. This is in contrast to the experimental design model in Section \ref{sec:dynamic} 
which chooses the optimal policy based on the whole experimental design phase. 
The greedy strategy is not theoretically optimal, as it essentially corresponds to truncating the sum in the dynamic 
programming formulation \eqref{eq:dyn_prog}. However, due to this truncation, 
the greedy approach does not suffer from the curse of dimensionality. 
Hence, it is more tractable from a computational point of view.

\section{Approximating the measure of residual uncertainty}
\label{sec:approx_H}
Assume $k$ experiments have been performed, resulting in the updated probabilistic model $(\xi_k, \delta_k)$.
A simple method for estimating the measures of residual uncertainty described in \secref{sec:opt_obj}, 
is by a double-loop Monte Carlo simulation: Let $N_1, N_2 \in \mathbb{N}$ and let $h^{(k)}_{i, j} = \ind{\xi_{k, j}(\x_i) \leq 0}$, 
where $\x_1, \dots, \x_{N_1}$ are $N_1$ i.i.d. samples of $\X$ and $\xi_{k, 1}(\x_i), \dots \xi_{k, N_2}(\x_i)$ are $N_2$ i.i.d. 
performance functions sampled from $\xi_k$ and evaluated at each $\x_i$. 
Then $H_{1, k}$ can be obtained as the sample variance of the $N_{2}$ samples of the form $\hat{\alpha}_{k, j} = \frac{1}{N_1}\sum_i h^{(k)}_{i, j}$.
Similarly, $H_{2, k}$ and $H_{3, k}$ can be estimated from $\hat{p}_k(\x_i) = \frac{1}{N_2}\sum_j h^{(k)}_{i, j}$.

This approach is problematic for several reasons. First of all, $\hat{\alpha}_{k, j}$ is an unbiased estimator 
of the failure probability $\alpha_{k, j} = \alpha(\xi_{k, j})$ corresponding to the deterministic performance function $\xi_{k, j}$.  
When $\alpha_{k, j}$ is small, the variance of this estimator is $\text{var}(\hat{\alpha}_{k, j}) = \alpha_{k, j}(1-\alpha_{k, j}) / N_1 \approx \alpha_{k, j} / N_1$. 
If we want to achieve an accuracy, of say $\sqrt{\text{var}(\hat{\alpha}_{k, j})} < 0.1 \alpha_{k, j}$, and 
$\alpha_{k, j} = 10^{-m}$, then the number of samples required would be approximately 
$N_1 = 10^{m + 2}$. The failure probabilities considered in structural reliability analysis can 
typically be in the range from $10^{-6}$ to $10^{-2}$. 

When $N_1$ is large, it can also be a practical challenge to obtain the samples $\xi_{k, j}(\x_1), \dots, \xi_{k, j}(\x_{N_1})$
simultaneously for a fixed $j$. Moreover, the total number of samples needed to evaluate the measures of residual uncertainty $H_{i, k}$ 
is $N_{1} N_{2}$, and we are interested in optimization over $H_{i, k}$ that will require multiple simulations of this kind. 

In this section we present a procedure for efficient approximation of the measures of residual uncertainty. 
We will start by introducing a finite dimensional approximation of $\xi_k(\x)$, given as a deterministic function 
$\hat{\xi}_k(\x, \E)$ depending on $\x$ and a finite dimensional $\mathcal{E}_k$-measurable random variable $\E$.
Then, in \secref{sec:UT_epistemic} we consider how the mean and variance, $\expect{f(\E)}$ and $\text{var}(f(\E))$, 
can be approximated for any $\mathcal{E}_k$-measurable function $f(\e)$ using the unscented transform.
In \secref{sec:sampling} and \secref{sec:IS_pruning} we present an importance sampling scheme for the case where $f(\e)$ is defined in terms of an expectation over $\X$.
Finally, in \secref{sec:UT_MCIS_approx} we consider the case where $f(\e) = \alpha(\hat{\xi}_k(\X, \e))$, 
which provides the approximations $\hat{\alpha}_k = f(\E)$ and $\hat{H}_{1, k} = \text{var}(f(\E))$, 
and where approximations of $H_{2, k}$ and $H_{3, k}$ are obtained in a similar manner. 

In summary, this kind of approximation which we will refer to as UT-MCIS from now on, makes use of the unscented transform (UT)
for epistemic uncertainty propagation and Monte Carlo simulation with importance sampling (MCIS)
for aleatory uncertainty propagation. The motivation behind this specific setup is that 
a technique such as MCIS is needed to obtain low variance estimates of $\alpha(\hat{\xi}_k(\X, \e))$, 
which will typically be a small number. The sampling scheme we propose is also designed 
to be efficient in the case where subsequent estimates corresponding to perturbations of $\alpha(\hat{\xi}_k(\X, \e))$ 
are needed, which is relevant for estimation of e.g. $\alpha(\hat{\xi}_{k+1}(\X, \e))$ or $\alpha(\hat{\xi}_k(\X, \e'))$ 
for some $\e' \neq \e$ if $\alpha(\hat{\xi}_k(\X, \e))$ has already been estimated. 
As for epistemic uncertainty propagation, when $\alpha(\hat{\xi}_k(\x, \E))$ is viewed as an 
$\mathcal{E}_k$-measurable random variable, 
the UT alternative which is both simpler and more efficient seems like a viable alternative, 
in particular for the purpose of optimization with respect to future decisions. 

\subsection{The finite-dimensional approximation of $\xi_k$}
\label{sec:finite_dim_approx}
In our framework, we have defined $\xi_k$ as a $\mathcal{E}_k$-measurable stochastic process 
indexed by $\x \in \mathbb{X}$ (often called a random field), and we view $\xi_k$ as a distribution over some (generally infinite-dimensional) space of functions.
The special case where $\xi_k = \xi_k(\x, \E)$ for some finite dimensional $\mathcal{E}_k$-measurable random variable $\E$ can be very useful for simulation.
That is, if samples $\e_j$ of $\E$ can be generated efficiently, then random functions $\xi_{k, j}(\x) = \xi_k(\x, \e_j)$ 
can be sampled as well. As long as $\xi_k$ is square integrable, such a representation 
of $\xi_k$ is always available from the Karhunen-Lo\'eve transform:
\begin{equation*}
    \xi_k(\x) - \expect{\xi_k(\x)} = \sum_{i = 1}^{\infty} E_i \phi_{i}(\x),
\end{equation*}
where the functions $\phi_{i}$ are deterministic and $E_i$ are uncorrelated random variables with zero mean. 
The canonical ordering of the terms $E_i \phi_{i}(\x)$ also provides a suitable method for approximating $\xi_k(\x)$, 
by truncating the sum at some finite $i = M$, and we could then let $\E = (E_1, \dots, E_M)$ (see for instance \cite{phd:Wang}). 

But obtaining the Karhunen-Lo\'eve transform can also be challenging. Because of this, we present 
an extremely simple approximation, that just relies on computation of the first two moments of $\xi_k$. 
We let $\E$ be a $1$-dimensional random variable with $\expect{\E} = 0$ and $\expect{\E^2} = 1$, and define
\begin{equation}
    \label{eq:finite_dim_xi}
    \hat{\xi_k}(\x) = \expect{\xi_k(\x)} + \E \sqrt{\text{var}(\xi_k(\x))}.
\end{equation}
This is indeed a very crude approximation, as essentially we assume that the values of $\xi_k$ at any set 
of inputs $\x$ are fully correlated. 
But for probabilistic surrogates used in structural reliability models 
this is actually not that unreasonable, and as it turns out, for the examples we consider in 
\secref{sec:num_exp} it seems sufficient. 
\begin{figure}[H]
    \centering
    \includegraphics[width=1.0\linewidth, trim={0cm 0cm 0cm 0cm},clip]{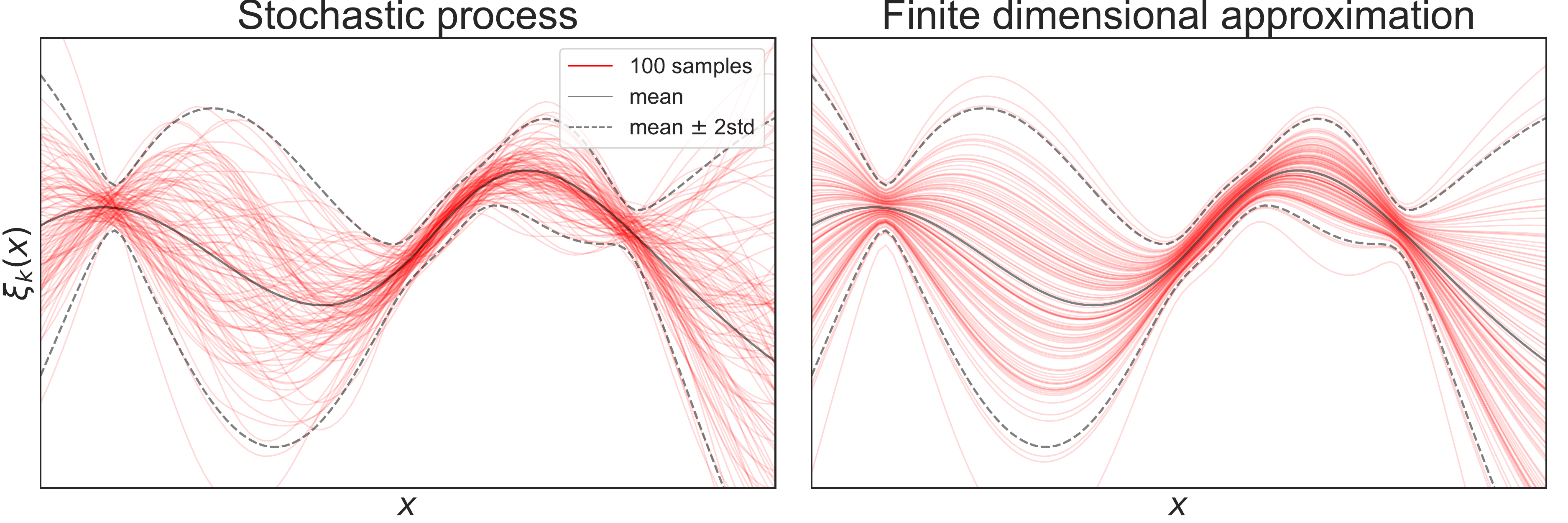}
    \caption{
        Illustration of the finite dimensional approximation \eqnref{eq:finite_dim_xi}.
    }
    \label{fig:example_finite_dim}
\end{figure}

\begin{remark}
    Note that to update the approximate model $\hat{\xi_k}(\x)$ in \eqnref{eq:finite_dim_xi} given some new 
    experiment $(d_k, o_k)$, we only need to update the mean and variance functions.
    This is in line with the numerically efficient Bayes linear approach \citep{Goldstein:2007:Bayes_linear_statistics},
    where random variables are specified only through the first two moments, and where the Bayesian updating given some
    experiment corresponds to computation of an adjusted mean and covariance. 
    An application of the Bayes linear theory to sequential optimal design of experiments 
    can be found in \citep{Jones:2018:Bayes_linear_design}.
    
    We note also that in the case where Gaussian processes are used as surrogate models, 
    the classical and linear Bayesian approaches are computationally equivalent. 
    Moreover, in the following section we will introduce the unscented transform for approximation of the
    updated/adjusted moments, and as a consequence the complete prior probability specification of $\E$
    becomes less relevant. 
\end{remark}

In the case where we are dealing with a hierarchical model, it might not be 
convenient to compute $\expect{\xi_k(\x)}$ and $\text{var}(\xi_k(\x))$. 
If $\xi_k(\x) = g(\Y_k(\x))$ where $\Y_k(\x)$ is a stochastic process with 
values in $\mathbb{R}^n$ for any $\x \in \mathbb{X}$, we would instead approximate $\Y_k$ with
\begin{equation}
    \hat{\Y}_k = \expect{\Y_k} + L \E, 
\end{equation}
where $\E$ is $n$-dimensional with $\expect{\E} = 0$, $E[\E \E^T] = I$, and 
the matrix $L$ satisfies $L L^T = (\Y_k-\expect{\Y_k})(\Y_k-\expect{\Y_k})^T$.
The approximation of $\xi_k$ is then obtained as $\hat{\xi_k}(\x) = g(\hat{\Y}_k(\x))$.
The same goes for the scenario with more than two layers in the hierarchy, 
for instance $\xi_k(\x) = g(\textbf{Z}_k(\Y_k(\x)))$, 
where we would approximate both $\textbf{Z}_k(\y)$ and $\Y_k(\x)$. In any case, we end up 
with finite dimensional random variable $\E$, and we can define the approximation $\hat{\xi_k}(\x, \E)$.

\subsection{The unscented transform for epistemic uncertainty propagation}
\label{sec:UT_epistemic}
The unscented transform (UT) is a very efficient method for approximating the mean and covariance of a random variable after 
nonlinear transformation. UT is commonly applied in the context of Kalman filtering, 
and it is based on the general idea that \emph{it is easier to approximate a probability distribution than an arbitrary nonlinear transformation} \citep{phd:Uhlmann, Julier:2004:UT}. 
Intuitively, given any finite-dimensional random variable $\E$ we may define a set of \textit{weighted sigma-points} $\{ (v_{i}, \e_{i}) \}$, 
such that if $\{ (v_{i}, \e_{i}) \}$ was considered as a discrete probability distribution, then its mean and covariance would coincide with $\E$.
For any nonlinear transformation $\Y = f(\E)$, if $\E$ was discrete we could compute the mean and covariance of $\Y$ exactly. The UT approximation 
is the result of such computation, where we make use of a small set of weighted points 
$\{ (v_{i}, \e_{i}) \}$. 

Specifically, let $\E$ be a finite dimensional random variable with mean $\bm{\mu}$ and covariance matrix $\bm{\Sigma}$.
A set of sigma-points for $\E$ is a set of weighted samples $\{ (v_1, \e_1),$ $\dots, (v_n, \e_n) \}$ such that 
\begin{equation}
    \label{eq:UT_sigma_points}
    \begin{split}
        &\bm{\mu} = \sum_{i=1}^{n} v_i \e_i, \\ 
        &\bm{\Sigma} = \sum_{i=1}^{n} v_i (\e_i - \bm{\mu})(\e_i - \bm{\mu})^T .
    \end{split}
\end{equation}
If $\y = f(\e)$ is any (generally nonlinear) transformation, the  
UT approximation of the mean and covariance of $\Y = f(\E)$ are then obtained as 
\begin{equation}
    \label{eq:UT_moments}
    \begin{split}
        &\widehat{E}[\Y] = \sum_{i=1}^{n} v_i \y_i, \\ 
        &\widehat{\text{Cov}}[\Y] = \sum_{i=1}^{n} v_i (\y_i - \widehat{E}[\Y])(\y_i - \widehat{E}[\Y])^T,
    \end{split}
\end{equation}
where $\y_i = f(\e_i)$.

Naturally, the selection of appropriate sigma-points is essential for UT to be successful. 
It is important to note that, although we may view the sigma-points as \emph{weighted samples}, 
$v_i$ and $\e_i$ are fixed or given by some deterministic procedure.
Moreover, the definition of sigma-points given in \eqnref{eq:UT_sigma_points} does not 
require that the weights are nonnegative and sum to one. 
Although this conflicts with the intuition of approximating $\E$ with a discrete random variable, 
the unscented transform still makes sense as a procedure for approximating statistics after nonlinear transformation. 

Since the introduction of UT to Kalman filters in the 1990's, many different alternatives for sigma-point selection have
been proposed \citep{Menegaz:2015:Systematization_UT}.
These are mostly focus on applications where $\E$ follows a multivariate Gaussian distribution, 
but we do not see this as a restriction since we will assume that $\E$ can be represented as
a transformation $\E = \mathcal{T}^{-1}(\bm{U})$ of a multivariate Gaussian variable $\bm{U}$.
For the applications considered in this paper, we will let $\{ (v_{i}, \uu_{i}) \}$ denote 
a set of sigma-points that are appropriate for the multivariate standard normal $\bm{U} \sim \mathcal{N}(0, I)$ 
where $\text{dim}(\bm{U}) = \text{dim}(\E)$. If $\mathcal{T}$ is the corresponding isoprobabilistic 
transformation, i.e. $\mathcal{T}(\E) \sim \mathcal{N}(0, I)$ (see \appref{sec:app:local_SRA}), we will use  $\{ (v_{i}, \mathcal{T}^{-1}(\uu_{i})) \}$ as a 
set of sigma-points for $\E$. Equivalently, we could also view this as taking the UT approximation of 
$\bm{U}$ under a different transformation given by $f \circ \mathcal{T}$.
For the numerical examples we present in this paper, we have made use of the 
the method developed by \cite{phd:Merwe}, which produces 
a set of $n = 2\cdot \text{dim}(\E) + 1$ points $\e_i$ with corresponding weights\footnote{
    Other alternatives for sigma-point selection could also be applied, potentially with better performance. 
    The method by \cite{phd:Merwe} depends on a set of parameters, 
    and it could also be relevant to refine or \emph{learn} the appropriate parameter values as in \citep{Turner:2010:learn_sigma_points}.
    However, in our current implementation we have only considered the fixed set of sigma-points 
    given in \appref{app:sigma_points}.
}.
Determining sigma-points with this procedure is quite straightforward, and the details are 
given in \appref{app:sigma_points}. We note again that for any structural reliability model, 
as long as we do not change dimensionality of $\E$, determining the sigma-points is a one-time computation, 
and any subsequent UT approximation of $\Y = f(\E)$, for some nonlinear transformation $f(\cdot)$, 
is computationally very efficient. 

\begin{remark}
    Note that it is not necessary that the sigma points used in the approximation of the mean and covariance in \eqnref{eq:UT_moments} are the same. In fact, the method presented in \appref{app:sigma_points} makes use of two different sets of weights for these approximations. 
    As this is not of any relevance for the remaining part of this paper, we will keep writing
    $\{ v_i, \e_i \}$ as a single set of sigma-points to simplify the notation. 
\end{remark}

\subsection{Generating samples in $\mathbb{X}$}
\label{sec:sampling}
In order to estimate the measures of residual uncertainty, we will need a set of samples of $\X$. 
We will generate a finite set of $3$-tuples $\{ (\x_i, w_i, \hat{\eta}_i) \}$, where $\{ (\x_i, w_i) \}$ are weighted samples in $\mathbb{X}$ 
suitable for obtaining importance sampling estimates of failure probabilities, and $\hat{\eta}_i$ is a number describing 
how influential a given sample $(\x_i, w_i)$ is expected to be in such an estimate. 
In other words, $\{ \x_i \}$ should be constructed to 'cover the relevant regions in $\mathbb{X}$', 
and for estimation we will only make use of a subset of $\{ (\x_i, w_i) \}$. 
The relevant subset will be determined from the \emph{measure of insignificance} $|\hat{\eta}_i|$, 
where we will only consider samples $(\x_i, w_i)$ where $|\hat{\eta}_i|$ is below some threshold.
We start by describing how the weighted samples $\{ (\x_i, w_i) \}$ are generated.\\

\noindent \textbf{Importance sampling}\\
\noindent The general idea behind importance sampling is that if we select some random variable $\bm{Q} \geq 0$ 
with law $P_{\bm{Q}}$, such that $E_{P_{\X}}[Q] = 1$ and $\bm{Q} \neq 0$ $P_{\X}$-almost surely, then 
\begin{equation}
    \label{eq:IS}
    E_{P_{\X}}[f(\X)] = E_{P_{\bm{Q}}}[ f(\X) / \bm{Q} ],
\end{equation}
for any $\mathcal{A}$-measurable function $f(\x)$. 
This is often useful for estimation, for instance when sampling from $P_{\X}$ is difficult, 
and in the case where we can find a $\bm{Q}$ such that estimates with respect to the right 
hand side of \eqnref{eq:IS} are better (have lower variance) than estimating $E_{P_{\X}}[f(\X)]$ directly.

In the case where $\X$ admits a probability density $p_{\X}$, we can let $q_{\X}$ be any 
density function such that $q_{\X}(\x) > 0$ whenever $p_{\X}(\x) > 0$. 
Let $\x_1, \dots, x_N$ be i.i.d. samples generated according to $q_{\X}$, 
and define $w_i = p_{\X}(\x_i) / q_{\X}(\x_i)$. The importance sampling estimate of 
$E_{P_{\X}}[f(\X)]$ with respect to the proposal density $q_{\X}$ is then obtained as 
\begin{equation}
    \label{eq:MCIS}
    \begin{split}
        E_{P_{\X}}[f(\X)] &= E_{P_{\bm{Q}}}\left[ f(\X)\frac{p_{\X}(\X)}{q_{\X}(\X)} \right]\\
        &\approx \frac{1}{N} \sum_{i = 1}^{N} f(\x_i) w_i .
    \end{split}
\end{equation}
We now assume that the stochastic limit state can be written as $\xi_k(\x, \E)$ for some 
finite-dimensional random variable $\E$, and for any deterministic performance function $\xi_k(\x, \e)$ 
we will write $\alpha_k(\e) = \alpha(\xi_k(\X, \e))$ as the corresponding failure probability.
An importance sampling estimate of $\alpha_k(\e)$ is then given by \eqnref{eq:MCIS} 
with $f(\x) = \ind{\xi_k(\x, \e) \leq 0}$, that is
\begin{equation}
    \label{eq:MCIS_alpha_k}
    \hat{\alpha}_k (\e)= \frac{1}{N}\sum_{i = 1}^{N} \ind{\xi_k(\x_i, \e) \leq 0} w_i .
\end{equation}
In order to obtain a good estimate of $\alpha_k(\e)$,
we would like the proposal distribution $q_{\X}$
to produce samples such that there is an even balance between the samples where $\xi_k(\x, \e) \leq 0$ 
and $\xi_k(\x, \e) > 0$, where at the same time $p_{\X}$ is as large as possible. 
One way to achieve this is to generate samples in the vicinity of points on the 
surface $\xi_k(\x, \e) = 0$ with (locally) maximal density. 
A point with this property is called a \emph{design point}\footnote{
    The most common definition of a design point is that it is the point on 
    the limit state surface with maximal density after transformation to the standard normal space. 
    See \appref{sec:app:local_SRA}
    } 
or \emph{most probable failure point} in the structural reliability literature. 
We will let $q_{\X}$ represent a mixture of distributions, centered around different design points that
are appropriate for different values of $\e$. The full details are given in \appref{app:sampling_dist_q},
where we also describe a simpler alternative than can be used in the case where design point searching
is difficult or not appropriate.\\

\noindent \textbf{The measure of insignificance} $| \eta_i |$\\
\noindent Assume $\{ (\x_i, w_i) \}$ is a set of samples capable of providing 
a satisfactory estimate of $\alpha_k(\e)$, and we now want to estimate $\alpha_k(\e')$ for some 
new value $\e'$. If we know that the sign of $\xi_k(\x_i, \e)$ and $\xi_k(\x_i, \e')$ will coincide 
for many of the samples $\x_i$, then the estimate of $\alpha_k(\e')$ can be obtained more efficiently by
not computing all the terms in the sum \eqnref{eq:MCIS_alpha_k}. 
This is typically the case when $\e$ and $\e'$ are both sampled from $\E$. 
It is also true in the case where we want to estimate $\alpha_{k+1}(\e')$ given some new experiment $(d_k, o_k)$, 
if we assume that updating with respect to $(d_k, o_k)$ has local effect (i.e. there are always 
regions in $\mathbb{X}$ where $\xi_{k+1}(\x) \approx \xi_k(\x)$), or if the experiment is 
carried out to reduce the uncertainty in the level set $\xi_k = 0$ (which is what we intend to do).

In other words, we consider some perturbation of the performance function $\xi_k(\x, \e)$, 
and we are interested in identifying the samples $\x_i$ where $\ind{\xi_k(\x_i, \e) \leq 0}$ 
does not change under the perturbation. For this purpose we define the function
\begin{equation}
    \label{eq:eta_def}
    \eta(\x, \xi) = \expect{\xi(\x)} / \sqrt{\text{var}(\xi(\x))},
\end{equation}
and let $\eta_i = \eta(\x_i, \xi_k)$ be defined with respect to the relevant process $\xi_k$.
Here $\eta_i$ describes how uncertain $\xi_k(\x_i)$ is around the critical value $\xi_k = 0$, 
in the sense that if $|\eta_i|$ is small (close to zero) then $\xi_k(\x_i) > 0$ and $\xi_k(\x_i) \leq 0$ may both be probable outcomes. 
Conversely, if $|\eta_i|$ is large then either $P(\xi_k(\x_i) \leq 0) \approx 0$ or $P(\xi_k(\x_i) \leq 0) \approx 1$, 
and the input $\x_i$ is \emph{insignificant} as it is unnecessary to keep track of changes in $\ind{\xi_k(\x_i) \leq 0}$. 
We will use $\eta_i$ to \emph{prune} the sample set $\{ (\x_i, w_i) \}$, by only considering the samples 
where $|\eta_i|$ is below a given threshold $\tau$. 
Although this is an intuitive idea, we may also justify the definition of $\eta$ and 
selection of a threshold $\tau$ more formally by making use of the following proposition.

\begin{prop}
    \label{prop:chebyshev}
    Given any process $\xi(\x)$, let $\eta(\x) = \eta(\x, \xi)$ be defined as in \eqref{eq:eta_def} 
    and let $\tau > \sqrt{2}$. Assume $\xi^{(1)}$ and $\xi^{(2)}$ are two i.i.d. random samples from $\xi(\x)$.
    Then, 
    \begin{equation}
        \label{eq:chebyshev}
        \begin{split}
            &P\left(\ind{\xi^{(1)} \leq 0} \neq \ind{\xi^{(2)} \leq 0} \;\middle|\; |\eta| \geq \tau \right) \\ &\leq \frac{2}{\tau^2} \left(1 - \frac{1}{\tau^2}\right).
        \end{split}
    \end{equation}
\end{prop}

\begin{proof}
    Let $p = P(\xi(\x) \leq 0)$ and $\gamma(p) = p(1-p)$ for short (note also that this is \eqnref{eq:pk_lambda_k} for $\xi = \xi_k$), 
    and observe that $P\left(\ind{\xi^{(1)} \leq 0} \neq \ind{\xi^{(2)} \leq 0}\right) = 2\gamma(p)$.
    Assume first that $\eta > 0$. Then $E[\xi] > 0$ and by Chebyshev's one-sided inequality we 
    get 
    \begin{equation*}
        \eta = \tau \Rightarrow p \leq \frac{\text{var}(\xi(\x))}{(\text{var}(\xi(\x)) + E[\xi(\x)]^2)} \leq \frac{1}{\tau^2},
    \end{equation*}
    and as $\tau > \sqrt{2}$ we also get $p \leq 1/2$. Since $\gamma(p)$ is increasing for $p \in [0, 1/2]$, 
    we must have $\gamma(p) \leq \gamma(1 / \tau^2)$.

    Conversely, if $-\tau = \eta < 0$ then $p \geq 1 - 1/ \tau^2 \geq 1/2$, and as $\gamma(p)$ is decreasing for $p \in [1/2, 1]$
    we have that $\gamma(p) \leq \gamma(1 - 1/ \tau^2) = \gamma(1/ \tau^2)$. 
    Hence, combining both cases we get $|\eta| = \tau \Rightarrow \gamma(p) \leq \gamma(1/ \tau^2)$, 
    and \eqnref{eq:chebyshev} is proved by observing that $\gamma(1 / (\tau + \varepsilon)^2) \leq \gamma(1/ \tau^2)$
    for any $\varepsilon > 0$.
    \qed
\end{proof}

\noindent Although \propref{prop:chebyshev} holds in general, tighter (and probably more realistic) bounds 
can be obtained by making assumptions on the form of $\xi(\x)$. 
For instance, in the case where $\xi(\x)$ is Gaussian we obtain 
\begin{equation}
\begin{split}
    &P\left(\ind{\xi^{(1)} \leq 0} \neq \ind{\xi^{(2)} \leq 0} \;\middle|\; |\eta| \geq \tau \right) \\
    &\leq 2\Phi(\tau)\Phi(-\tau),
\end{split}
\end{equation}
where $\Phi(\cdot)$ is the standard normal CDF.

We will make use of $\hat{\eta}_i$ obtained as the UT approximation of $\eta_i$.
That is, $\hat{\eta}_i$ is in general obtained from the finite-dimensional approximation described in \secref{sec:finite_dim_approx},
combined with the UT approximation \eqnref{eq:UT_moments} with $\Y = \hat{\xi}_k(\x, \E)$.

\subsection{Importance sampling estimates with pruning}
\label{sec:IS_pruning}
Let $\{ (\x_i, w_i, \hat{\eta}_i) \ | \ i \in \mathcal{I} \}$, $\mathcal{I} = \{ 1, \dots, N_0 \}$ be a set of samples generated as described in \secref{sec:sampling}. 
Given some fixed threshold $\tau > 0$, we define the subset of \emph{pruned samples} as the ones corresponding to the 
index set $\mathcal{I}_{\tau} = \{ i \in \mathcal{I} \ | \  \hat{\eta}_i < \tau \}$, and define $\bar{\mathcal{I}}_{\tau} = \mathcal{I} \setminus \mathcal{I}_{\tau}$. 
If $f(\x)$ is some $\mathcal{A}$-measurable function where we know a priori the value of $f_i = f(\x_i)$ for all $i \in \bar{\mathcal{I}}_{\tau}$, 
then we can immediately compute
\begin{equation}
    \label{eq:h_bar}
    \bar{h} = \frac{1}{N_0} \sum_{i \in \bar{\mathcal{I}}_{\tau}} f_i w_i,
\end{equation}
and the importance sampling estimate of the expectation of $f(\X)$ becomes
\begin{equation}
    \widehat{E}[f(\X)] = \bar{h} + \frac{1}{N_0} \sum_{i \in \mathcal{I}_{\tau}} f(\x_i)w_i.
\end{equation}
If we let
\begin{equation}
    \label{eq:s_h_bar}
    s_{\bar{h}} = \frac{1}{N_{0}} \sum_{i \in \bar{\mathcal{I}}_{\tau}} \left( f_i w_i - \widehat{E}[f(\X)] \right)^2,
\end{equation}
then an unbiased estimate of the sample variance is given as 
\begin{equation}
\begin{split}
    &\widehat{\text{var}}(\widehat{E}[f(\X)]) = \frac{s_{\bar{h}}}{N_0 - 1} \\
    &+ \frac{1}{N_{0}(N_{0}-1)} \sum_{i \in \mathcal{I}_{\tau}} \left( f(\x_i)w_i - \widehat{E}[f(\X)] \right)^2,
\end{split}
\end{equation}
which shows the general idea with this \emph{pruning}, namely that low variance estimates of $E[f(\X)]$ can be 
obtained with a small number of evaluations $f(\x_i)$, assuming that the subset $\mathcal{I}_{\tau}$ is small compared to $\mathcal{I}$ 
(and that the assumed values $f_i$ are correct).

One drawback with this procedure is that we do not have control over the number of pruned samples, 
which still might be very large. In order to set an upper bound on the number of evaluations $f(\x_i)$, 
we let $\mathcal{I}_{\tau}^n \subseteq \mathcal{I}_{\tau}$ contain the first $n$ elements of $\mathcal{I}_{\tau}$ 
(or some other subset, as long as the elements of $\{\x_i \ | \ i \in  \mathcal{I}_{\tau}^n\}$ remain independent).
An importance sampling estimate of $E[f(\X)]$ using only samples from $\mathcal{I}_{\tau}^n$ is given as 
\begin{equation}
    \label{eq:split_MC}
    \widehat{E}[f(\X)] = \bar{h} + \bar{r}, \ \ \bar{r} = \frac{N_{\tau}}{nN_{0}} \sum_{i \in \mathcal{I}_{\tau}^{n}} f(\x_i)w_i,
\end{equation}
where $N_{\tau} = |\mathcal{I}_{\tau}|$, and we may estimate the sample variance as
\begin{equation}
\label{eq:split_MC_var}
\begin{split}
    &\widehat{\text{var}}(\widehat{E}[f(\X)]) = \frac{1}{N_{0}-1}(s_{\bar{h}} - \bar{h}^2) \\
    &+ \frac{N_{\tau}}{nN_{0} - N_{\tau}} \left( -\bar{r}^2 + \frac{N_{\tau}}{nN_{0}}\sum_{i \in \mathcal{I}_{\tau}^n} \left( f(\x_i)w_i \right)^2 \right).
\end{split}
\end{equation}
Obtaining consistency results is easy under the ideal assumption that $n(N_0 - N_{\tau}) / N_{\tau}$ is an integer, 
and the formulas in \eqref{eq:split_MC}-\eqref{eq:split_MC_var} comes as a consequence 
of the following result.
\begin{prop}
    Assume $n(N_0 - N_{\tau}) / N_{\tau} \in \mathbb{N}$. Then \eqnref{eq:split_MC} is an unbiased estimate 
    of $E[f(\X)]$ and \eqnref{eq:split_MC_var} is an unbiased estimate of the sample variance. 
\end{prop}
\begin{proof}
    Let $\bar{\mathcal{I}}_{\tau}^n$ be a set of $n(N_0 - N_{\tau}) / N_{\tau}$ elements selected uniformly random from $\bar{\mathcal{I}}_{\tau}$
    and define $\mathcal{I}^n = \mathcal{I}_{\tau}^n \cup \bar{\mathcal{I}}_{\tau}^n$.
    Then $\{\x_i \ | \ i \in  \mathcal{I}^n\}$ is a set of size $|\mathcal{I}^n| = n N_0 /  N_{\tau}$, 
    containing i.i.d. samples from the proposal distribution with density $q(\x)$.
    To show consistency we replace each sample $\x_i$ with i.i.d. random variables $\X_i$ distributed 
    according to $q$. We then define $\hat{\mu} = \hat{\mu}_1 + \hat{\mu}_2$ where  
    \begin{equation*}
    \begin{split}
        \hat{\mu}_1 &= \frac{1}{|\mathcal{I}|} \sum_{i \in \mathcal{I}} \ind{\eta(\X_i \geq \tau)}f(\X_i)w(\X_i), \\
        \hat{\mu}_2 &= \frac{1}{|\mathcal{I}^n|} \sum_{i \in \mathcal{I}^n} \ind{\eta(\X_i < \tau)}f(\X_i)w(\X_i), 
    \end{split}
    \end{equation*}
    and where $w(\x) = p(\x) / q(\x)$, and we can observe that $\hat{\mu} = \widehat{E}[f(\X)]$ when $\X_i = \x_i$.

    To show that $\widehat{E}[f(\X)]$ is unbiased it is enough to observe that 
    $E_q[\hat{\mu}] = E_q[\ind{\eta(\X \geq \tau)}f(\X)w(\X)] + E_q[\ind{\eta(\X < \tau)}f(\X)w(\X)]
    = E_q[f(\X)w(\X)] = E[f(\X)]$.

    As for the variance, we first observe that $\text{var}(\hat{\mu}) = \text{var}(\hat{\mu}_1) + \text{var}(\hat{\mu}_2)$
    where\\ $\text{var}(\hat{\mu}_1) = \text{var}(\ind{\eta(\X \geq \tau)}f(\X)w(\X)) / |\mathcal{I}|$ and  
    $\text{var}(\hat{\mu}_2) = \text{var}(\ind{\eta(\X < \tau)}f(\X)w(\X)) / |\mathcal{I}^n|$.
    Replacing $\text{var}(\hat{\mu}_1)$ and $\text{var}(\hat{\mu}_2)$ with unbiased sample variances 
    using the samples $\X_i = \x_i$ we obtain 
    \begin{equation*}
        \begin{split}
            &\widehat{\text{var}}(\hat{\mu}_1) =\\
            &\frac{1}{|\mathcal{I}|(|\mathcal{I}| - 1)} \sum_{i \in \mathcal{I}} \left(  \ind{\eta(\x_i \geq \tau)}f(\x_i)w(\x_i) - \bar{h} \right)^2 \\
            &= \frac{1}{|\mathcal{I}| - 1} \left(-\bar{h}^2 + \frac{1}{|\mathcal{I}|}\sum_{i \in \mathcal{I}} \left(  \ind{\eta(\x_i \geq \tau)}f(\x_i)w(\x_i) \right)^2 \right) \\ 
            &= \frac{1}{|\mathcal{I}| - 1}(-\bar{h}^2 + s_{\bar{h}}),
        \end{split}
    \end{equation*}
    and similarly 
    \begin{equation*}
        \widehat{\text{var}}(\hat{\mu}_2) = \frac{1}{|\mathcal{I}^n| - 1} \left(-\bar{r}^2 + \frac{1}{|\mathcal{I}^n|}\sum_{i \in \mathcal{I}_{\tau}^n} \left( f(\x_i)w(\x_i) \right)^2 \right),
    \end{equation*}
    where we have used that $\bar{h}$ and $\bar{r}$ are unbiased estimates of $E_q[\hat{\mu}_1]$ and $E_q[\hat{\mu}_2]$ respectively.
    The expression in \eqnref{eq:split_MC_var} is then obtained as $\widehat{\text{var}}(\hat{\mu}_1) + \widehat{\text{var}}(\hat{\mu}_2)$ using that $|\mathcal{I}| = N_0$ and $|\mathcal{I}^n| = n N_0 /  N_{\tau}$.
    \qed
\end{proof}

\subsection{The UT-MCIS approximation of $H_{1, k}$, $H_{2, k}$ and $H_{3, k}$}
\label{sec:UT_MCIS_approx}
Using the tools introduced in the preceding subsections, we now present how 
the measures of residual uncertainty, $H_{1, k}$, $H_{2, k}$ and $H_{3, k}$, can be approximated 
using Monte Carlo simulation with importance sampling (MCIS) 
combined with the unscented transform (UT) for epistemic uncertainty propagation.  

We first let $\hat{\xi_k}(\x, \E)$ be the finite-dimensional approximation introduced in \secref{sec:finite_dim_approx}, 
with the corresponding failure probability $\hat{\alpha}_{k}(\E) = \alpha(\hat{\xi_k}(\x, \E))$.
We then let $\{ (\x_i, w_i, \hat{\eta}_i) \ | \ i \in \mathcal{I} \}$, $\mathcal{I} = \{ 1, \dots, N_0 \}$ be a set of samples generated as described in \secref{sec:sampling}, 
where $\hat{\eta}_i$ is obtained using the UT approximation of $\hat{\xi_k}(\x_i, \E)$.
We will make use of importance sampling estimates as introduced in \secref{sec:IS_pruning}, 
where $\mathcal{I}_{\tau} = \{ i \in \mathcal{I} \ | \  \hat{\eta}_i < \tau \}$, 
and estimation is based on a small subset $\{ (\x_i, w_i, \hat{\eta}_i) \ | \ i \in \mathcal{I}_{\tau}^n \}$
where $\mathcal{I}_{\tau}^n \subset \mathcal{I}_{\tau}$ and $| \mathcal{I}_{\tau}^n | = n < N_{\tau} = |\mathcal{I}_{\tau}|$. \\

\noindent \textbf{Approximating} $\bm{H_{1, k}}$\\
\noindent Let $f_i = \ind{\hat{\eta}_i \leq 0}$ for $i \in \bar{\mathcal{I}}_{\tau}$ 
and compute $\bar{h}_{1}$ as in \eqnref{eq:h_bar}. 
We will let $\{ (v_{j}, \e_{j}) \ | \ j = 1, \dots, M \}$ denote the set of sigma-points as introduced in \secref{sec:UT_epistemic}.

For any fixed $\e_j$, the corresponding importance sampling estimate of the failure probability $\hat{\alpha}_k(\e_j)$ is obtained as 
\begin{equation}
    \label{eq:H1k_approx_1}
    \hat{\alpha}_k^j = \bar{h}_{1} + \frac{N_{\tau}}{nN_{0}} \sum_{i \in \mathcal{I}_{\tau}^{n}}
    \ind{\hat{\xi_k}(\x_i, \e_j) \leq 0} w_i,
\end{equation}
and we let $\hat{H}_{1, k}$ be given by the UT approximation 
\begin{equation}
    \label{eq:H1k_approx_2}
    \begin{split}
        &\widehat{E}[\hat{\alpha}_k] = \sum_{j = 1}^{M} v_j \hat{\alpha}_k^j, \\ 
        &\hat{H}_{1, k} = \widehat{\text{var}}[\hat{\alpha}_k] = \sum_{j = 1}^{M} v_j (\hat{\alpha}_k^j - \widehat{E}[\hat{\alpha}_k])^2.
    \end{split}
\end{equation}

\noindent \textbf{Approximating} $\bm{H_{2, k}}$ \textbf{and} $\bm{H_{3, k}}$\\
\noindent Both $H_{2, k}$ and $H_{3, k}$ are defined through the function $\gamma_k(\x)$, 
which represents the uncertainty in the sign of $\xi_k(\x)$. We will approximate $\gamma_k(\x_i)$ 
with the following function 
\begin{equation}
    \label{eq:gamma_approx}
    \hat{\gamma}_{k}^i = \Phi(\hat{\eta}_i)\Phi(-\hat{\eta}_i),
\end{equation}
where $\Phi(\cdot)$ is the standard normal CDF. There are two ways of interpreting this approximation. 
First of all, $\hat{\gamma}_{k, i}$ corresponds to the case where $\hat{\xi_k}(\x_i, \E)$ is Gaussian, 
which may or may not be an appropriate assumption. 
Alternatively, we can think of $\gamma_k(\x)$ as a measure of uncertainty in $\ind{\xi_k(\x) \leq 0}$, 
and any $\gamma(\x) \propto - |\eta(\x)| = - |E[\xi_k(\x)]| / \sqrt{ \text{var}(\xi_k(\x))} $ is reasonable.
In this scenario it is natural to consider $\gamma = s(\eta)s(-\eta)$ for some sigmoid function $s(\cdot)$,  
and the function $\Phi(\cdot)$ in \eqnref{eq:gamma_approx} is one such alternative. 

For a single approximation of $H_{2, k}$ or $H_{3, k}$ it is really not necessary to 
split the importance sampling estimate as in \eqnref{eq:h_bar}-\eqnref{eq:split_MC}, 
but we will present it in this form as it will be convenient when we consider strategies for optimization. 
Given $\hat{\gamma}_{k}^i$ as in \eqnref{eq:gamma_approx}, we approximate $H_{2, k}$ and $H_{3, k}$ by 
\begin{equation}
    \begin{split}
        &\hat{H}_{2, k} = \bar{h}_{2} + \frac{N_{\tau}}{nN_{0}} \sum_{i \in \mathcal{I}_{\tau}^{n}} \hat{\gamma}_{k}^i w_i, \\ 
        &\hat{H}_{3, k} = \left( \bar{h}_{3} + \frac{N_{\tau}}{nN_{0}} \sum_{i \in \mathcal{I}_{\tau}^{n}} \sqrt{\hat{\gamma}_{k}^i} w_i \right)^2,
    \end{split}
\end{equation}
where we let $\bar{h}_{2} = \bar{h}_{3} = 0$. Alternatively, if the intention is to use $H_{2, k}$ and $H_{3, k}$ 
as upper bounds on $H_{1, k}$, we could let $\bar{h}_{2} = \frac{1}{N_0} \Phi(\tau)\Phi(-\tau) \sum w_i$, 
$\bar{h}_{2} = \frac{1}{N_0} \sqrt{\Phi(\tau)\Phi(-\tau)}$ $\sum w_i$ where the sums are over $i \in \bar{\mathcal{I}}_{\tau}$.

\section{Numerical procedure for myopic optimization}
\label{sec:myopic_num}
In the myopic scenario, the optimal decision $d_k$ at each time step $k$ is found by solving 
the following optimization problem 
\begin{equation}
    \label{eq:myopic_optimal}
    d_k = \argmin_{d \in \mathbb{D}} J_{i, k}(d) \textit{ for } i = 1, 2, \textit{or } 3,
\end{equation}
where $J_{i, k}(d)$ is the relevant acquisition function as defined in \eqnref{eq:acquisition}.
We propose a procedure where we make use of a UT-MCIS approximation of $J_{i, k}(d)$ to find 
an approximate solution to \eqnref{eq:myopic_optimal}. 
This will build on the approximation of $H_{i, k}$ introduced in \secref{sec:approx_H}, 
but where we now also make use of the predictive model $\delta$ to approximate 
expectations with respect to future values of $H_{i, k+1}$.

In \secref{sec:myopic_pobmodel} and \secref{sec:myopic_acq} we present how the UT-MCIS 
approximation of $J_{i, k}(d)$ is obtained, and in \secref{sec:myopic_conv} we propose
a criterion for determining when the sequence of experiments should be stopped. 
The final algorithm is summarized in \secref{sec:myopic_alg}

\subsection{The probabilistic model $(\hat{\xi}_{k}, \hat{\delta}_{k})$}
\label{sec:myopic_pobmodel}
Starting with some probabilistic model $(\xi_k, \delta_k)$, recall that $\xi_k$ represents uncertainty about 
the \emph{performance} of the system under consideration, and $\delta_k$ represents 
uncertainty with respect to \emph{outcomes} of certain \emph{decisions}.
We have already discusses how to obtain a finite-dimensional approximation of $\xi_k$, 
and likewise, this will also be needed for $\delta_k$.

Assuming $\delta_k$ is square integrable, we will make use of the same type of finite-dimensional 
approximation as the one introduced for $\xi_k$ in \secref{sec:finite_dim_approx}.
In this way, we end up with two finite-dimensional $\mathcal{E}_k$-measurable random variables $\E^{\xi}$ and 
$\E^{\delta}$, which in turn determine the 
approximations $\hat{\xi}_k(\x, \E^{\xi})$ and $\hat{\delta}_k(d, \E^{\delta})$, 
where both $\hat{\xi}_k(\x, \e)$ and $\hat{\delta}_k(d, \e)$ are deterministic functions for $\e$ fixed. 
Here $\E^{\xi}$ and $\E^{\delta}$ are generally not independent. 

\begin{remark}
    Note that if $\delta(d)$ is a function 
    of some of the uncertain sub-components of $\xi$, then we might already have a finite-dimensional approximation of 
    $\delta$ available. 
    
    Consider for instance the model in \exref{ex:obs_proc} and the discussion 
    in the end of \secref{sec:finite_dim_approx}. In this case, $\hat{\xi}$ is obtained as a function 
    of the finite-dimensional approximation $\hat{y}_1(\x, \E)$ of a sub-component $\widetilde{y_{1}}(\x)$, 
    and $\delta(d)$ is given as $\delta(d(\x)) = \widetilde{y_{1}}(\x) + \epsilon(\x)$. 
    Hence, all we need is to find a finite-dimensional representation of the noise $\epsilon(\x)$.
    But observational noise such as $\epsilon(\x)$ is often described as a function of $\x$ and 
    some $1$-dimensional random variable, in which case no additional approximation will be needed. 
\end{remark}

We will let $(\hat{\xi}_{k}, \hat{\delta}_{k})$ denote the finite-dimensional approximation of 
$(\xi_k, \delta_k)$ corresponding to a finite-dimensional random variable $\E = (\E^{\xi}, \E^{\delta})$, 
and where $(\hat{\xi}_{0}, \hat{\delta}_{0})$ is the initial model that is used as input 
for determining the first decision $d_1$.

\begin{remark}
    In the canonical case where a surrogate $\tilde{y}(\x)$ is used to represent 
    some unknown function $y(\x)$, an initial set of experiments is often performed to 
    establish $\tilde{y}(\x)$ before any sequential strategy is started. 
    For instance, in the case where evaluation of $y(\x)$ means running deterministic 
    computer code, it is normal to set up a space-filling initial design using 
    e.g. Latin Hypercube Sampling. 
    
    When $\tilde{y}(\x)$ is a Gaussian process model as described in \appref{app:GP}, 
    specific mean and covariance functions may also be selected based on knowledge or assumptions 
    about the phenomenon that is being modelled by $y(\x)$. 
    For estimation of failure probabilities it is also convenient to make use of conservative prior mean values. 
    That is, prior to any experiment $\tilde{y}(\x)$ will correspond to a value associated with poor structural performance (small $\xi$), 
    such that $\alpha(\xi)$ will be biased towards higher failure probabilities in the absence of experimental evidence. 
    This reasonable from a safety perspective, and also numerically as larger failure probabilities are easier to estimate. 
\end{remark}

\subsection{Acquisition function approximation}
\label{sec:myopic_acq}
To find an approximate solution to the optimization problem \eqnref{eq:myopic_optimal},
we will replace the acquisition function $J_{i, k}(d)$ with an approximation $\hat{J}_{i, k}(d)$.
Recall that $J_{i, k}(d)$ as defined in \eqnref{eq:acquisition} is a function of $\expectK{H_{i, k+1}}{k, d}$, where 
$E_{k, d}$ is the conditional expectation with respect to $\mathcal{E}_{k}$ with $d_k = d$.
In \secref{sec:approx_H} we introduced an approximation $H_{i, k}$, and we will make use of the same 
idea to approximate $\expectK{H_{i, k+1}}{k, d}$. 

Assume $k$ experiments have been performed, giving rise to the model $(\xi_{k}, \delta_{k})$ and the approximation $(\hat{\xi}_{k}, \hat{\delta}_{k})$.
If we consider the $k$-th decision $d_k = d$, then $H_{i, k+1}$ is a priori a $\delta_k(d)$-measurable random variable. 
That is, $H_{i, k+1}$ is a function of $\delta_k(d)$, 
and we are interested in the expectation $\expectK{H_{i, k+1}}{k, d} = \expect{H_{i, k+1}(\delta_k(d))}$. 
To approximate this quantity, we can make use of $(\hat{\xi}_{k}, \hat{\delta}_{k})$ in the place of $(\xi_k, \delta_k)$, 
in which case $H_{i, k+1}$ becomes a function of $\E$ and we can approximate its expectation using UT.

The approximate acquisition functions are then given as 
\begin{equation}
    \label{eq:approx_acq}
    \hat{J}_{i, k}(d) = \lambda(d) \widehat{E}_{{k, d}}[\hat{H}_{i, k+1}],
\end{equation}
where $\widehat{E}_{{k, d}}[\hat{H}_{i, k+1}]$ is obtained as follows:\\

\noindent \textbf{Generating samples of} $\bm{\hat{\xi}_{k+1}}$\\
\noindent Let $\{ (v_{j}^{\xi}, \e_{j}^{\xi}) \ | \ j = 1, \dots, M^{\xi} \}$ and $\{ (v_{m}^{\delta}, \e_{m}^{\delta}) \ | \ m = 1, \dots, M^{\delta} \}$ denote sigma-points as introduced in \secref{sec:UT_epistemic}
for $\E^{\xi}$ and $\E^{\delta}$ respectively. 
We then let $\{ (\x_i, w_i, $ $\hat{\eta}_i) \ | \ i \in \mathcal{I} \}$, $\mathcal{I} = \{ 1, \dots, N_0 \}$ be a set of samples generated as described in \secref{sec:sampling}, 
where $\hat{\eta}_i$ is obtained using the UT approximation of $\hat{\xi_k}(\x_i, \E^{\xi})$.
As for the approximation of $H_{i, k}$ discussed in \secref{sec:UT_MCIS_approx}, 
we let $\mathcal{I}_{\tau} = \{ (\x_i, w_i, \hat{\eta}_i) \ | \hat{\eta}_i < \tau \}$
and define the subset $\mathcal{I}_{\tau}^n \subseteq \mathcal{I}_{\tau}$ of size $n$.

The approximations of $\expectK{H_{i, k+1}}{k, d}$ for $i = 1, 2$ and $3$ will all be based on 
samples of $\hat{\xi}_{k+1}$ of the form 
\begin{equation}
    \label{eq:xi_k1_samples}
    \hat{\xi}^{m, i, j}_{k+1} = \hat{\xi}_{k+1}(\x, \e^{\xi}_j, d, \e^{\delta}_m),
\end{equation}
where $\hat{\xi}_{k+1}(\x, \e^{\xi}_j, d, \e^{\delta}_m)$ is the finite-dimensional approximation 
of $\xi_k | d_k = d, o_k = \hat{\delta}(\e^{\delta}_m)$ evaluated at $(\x, \e^{\xi}_j)$.
The scalar $\hat{\xi}^{m, i, j}_{k+1}$ is computed for all $j = 1, \dots, M^{\xi}$, $m = 1, \dots, M^{\delta}$
and $i = \in \mathcal{I}_{\tau}^n$.
As in \secref{sec:UT_MCIS_approx} we set $\bar{h}_2 = \bar{h}_3 = 0$ and compute $\bar{h}_{1}$ as in \eqnref{eq:h_bar} with $f_i = \ind{\hat{\eta}_i \leq 0}$ for $i \notin \mathcal{I}_{\tau}$.\\

\noindent \textbf{The UT-MCIS approximation of} $\bm{\expectK{H_{1, k+1}}{k, d}}$\\
\noindent The approximation $\widehat{E}_{{k, d}}[\hat{H}_{1, k+1}]$ is just a weighted sum 
of the terms in \eqnref{eq:xi_k1_samples}, but for clarity we present it in the following three steps
\begin{align}
&\text{MCIS of } \alpha(\hat{\xi}_{k+1}):\nonumber \\ & \hat{\alpha}^{m, j}_{k+1} = 
\bar{h}_{1} + \frac{N_{\tau}}{nN_{0}} \sum_{i \in \mathcal{I}_{\tau}^{n}}
    \ind{\hat{\xi}^{m, i, j}_{k+1} \leq 0} w_i, \\
&\text{UT of } H_{1, k+1}:\nonumber \\ &\hat{H}_{1, k+1}^{m} = \sum_{j=1}^{M^{\xi}} v_{j}^{\xi}
(\hat{\alpha}^{m, j}_{k+1})^2 - \left( \sum_{j=1}^{M^{\xi}} v_{j}^{\xi}
\hat{\alpha}^{m, j}_{k+1} \right)^2, \\
&\text{UT of } \expectK{H_{1, k+1}}{k, d}:\nonumber \\ &\widehat{E}_{{k, d}}[\hat{H}_{1, k+1}] = 
\sum_{m=1}^{M^{\delta}}  v_{m}^{\delta} \hat{H}_{1, k+1}^{m}. 
\label{eq:UT-MCIS-H1}
\end{align}\\

\noindent \textbf{The UT-MCIS approximation of} $\bm{\expectK{H_{2, k+1}}{k, d}}$ \textbf{and} $\bm{\expectK{H_{3, k+1}}{k, d}}$\\
\noindent The weighted sums that gives the approximations of $\expectK{H_{2, k+1}}{k, d}$ and 
$\expectK{H_{3, k+1}}{k, d}$ can be obtained as follows
\begin{align}
    &\text{UT of } E[\hat{\xi}_{k+1}(\x_i)]: \; \; \; \hat{\mu}^{i, m}_{k+1} = \sum_{j=1}^{M^{\xi}} v_{j}^{\xi} \hat{\xi}^{m, i, j}_{k+1}, \\
    &\text{UT of } \text{var}[\hat{\xi}_{k+1}(\x_i)]:\nonumber\\ &(\hat{\sigma}^{i, m}_{k+1} )^2 = \sum_{j=1}^{M^{\xi}} v_{j}^{\xi} (\hat{\xi}^{m, i, j}_{k+1} - \hat{\mu}^{i, m}_{k+1})^2, \\
    &\text{Using } \Phi \text{ to approximate } \hat{\gamma}_{k+1}(\xi_i):\nonumber\\ & \hat{\gamma}_{k+1}^{i, m} = 
    \Phi(\hat{\eta}_{k+1}^{i, m})\Phi(-\hat{\eta}_{k+1}^{i, m}), \ \  \hat{\eta}_{k+1}^{i, m} = \hat{\mu}^{i, m}_{k+1} / \hat{\sigma}^{i, m}_{k+1} \\[2ex]
    &\text{MCIS of } H_{2, k+1}:\nonumber\\ & \hat{H}_{2, k+1}^{m} =
    \bar{h}_{2} + \frac{N_{\tau}}{nN_{0}} \sum_{i \in \mathcal{I}_{\tau}^{n}} \hat{\gamma}_{k+1}^{i, m} w_i, \\
    &\text{MCIS of } H_{3, k+1}:\nonumber\\ & \sqrt{\hat{H}_{3, k+1}^{m}} =
    \bar{h}_{3} + \frac{N_{\tau}}{nN_{0}} \sum_{i \in \mathcal{I}_{\tau}^{n}} \sqrt{\hat{\gamma}_{k+1}^{i, m}} w_i, \\
\end{align}
and where $\widehat{E}_{{k, d}}[\hat{H}_{2, k+1}]$ and $\widehat{E}_{{k, d}}[\hat{H}_{3, k+1}]$ are 
obtained with the same formula as for $\widehat{E}_{{k, d}}[\hat{H}_{1, k+1}]$ in \eqref{eq:UT-MCIS-H1}.

\begin{remark}
    The number of model updates and function evaluations needed to generate the set $\{\hat{\xi}^{m, i, j}_{k+1}\}$ 
    are $M^{\delta}$ and $n M^{\xi} M^{\delta}$. We can view this as a discretization of the system dynamics, 
    where there are only $ M^{\delta}$ possible future scenarios corresponding to the decision $d_k = d$, 
    which are given by the model updates $\xi_k \rightarrow \xi_{k+1}(\e^{\delta}_m) = \xi_k | d_k = d, o_k = \hat{\delta}_k(\e^{\delta}_m)$. The samples in \eqnref{eq:xi_k1_samples} are the ones needed for 
    approximating the measure of residual uncertainty corresponding to $\xi_{k+1}(\e^{\delta}_m)$ for each 
    $m = 1, \dots, M^{\delta}$.
    
    Moreover, although the approximations $\widehat{E}_{{k, d}}[\hat{H}_{i, k+1}]$ are presented as 
    weighted sums of the $n M^{\xi} M^{\delta}$ terms $\hat{\xi}^{m, i, j}_{k+1}$, this can also 
    be obtained from a sequence of nested loops for a more memory efficient implementation. 
    See for instance the schematic illustration in \figref{fig:schematic_approx} below. 

    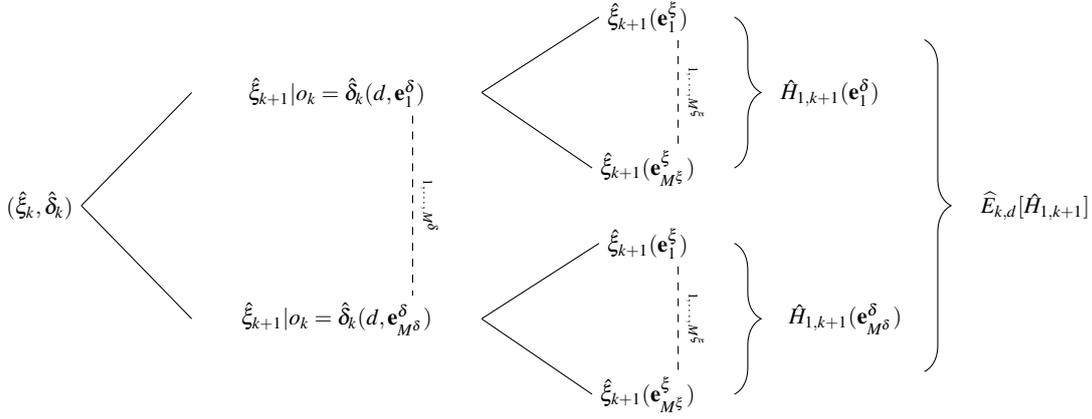
\begin{figure*}[ht]
        \centering
            \begin{tikzpicture}
                \tikzstyle{txtnode} = [rectangle, anchor=west, minimum width=3.8cm, inner sep=0];
                
                \node at (0, -0.5) (mk) {$(\hat{\xi}_k, \hat{\delta}_k)$ };
                \node[txtnode] at (2, 1) (mk1_1) {$\hat{\xi}_{k+1} | o_k = \hat{\delta}_k(d, \e_1^{\delta})$ };
                \node[txtnode] at (2, -2) (mk1_2) {$\hat{\xi}_{k+1} | o_k = \hat{\delta}_k(d, \e_{M^{\delta}}^{\delta})$ };
                \draw (mk.east) -- (mk1_1.west);
                \draw (mk.east) -- (mk1_2.west);
                
                \node[right of = mk1_1, yshift = -0.2cm] (tmp1) {};
                \node[right of = mk1_2, yshift = 0.2cm] (tmp2) {};
                \draw (tmp1) -- (tmp2) [dashed] node [midway, above, sloped] () {\tiny{$1, \dots, M^{\delta}$}};
                
                \node at (8, 2) (mk2_1) {$\hat{\xi}_{k+1}(\e^{\xi}_1)$ };
                \node at (8, 0) (mk2_2) {$\hat{\xi}_{k+1}(\e^{\xi}_{M^{\xi}})$ };
                \draw (mk1_1.east) -- (mk2_1.west);
                \draw (mk1_1.east) -- (mk2_2.west);
                
                \node at (8, -1) (mk2_3) {$\hat{\xi}_{k+1}(\e^{\xi}_1)$ };
                \node at (8, -3) (mk2_4) {$\hat{\xi}_{k+1}(\e^{\xi}_{M^{\xi}})$ };
                \draw (mk1_2.east) -- (mk2_3.west);
                \draw (mk1_2.east) -- (mk2_4.west);
                
                \node[right of = mk2_3, xshift = -0.6cm, yshift = -0.2cm] (tmp1) {};
                \node[right of = mk2_4, xshift = -0.6cm, yshift = 0.2cm] (tmp2) {};
                \draw (tmp1) -- (tmp2) [dashed] node [midway, above, sloped] () {\tiny{$1, \dots, M^{\xi}$}};
                \draw [decorate,decoration={brace,amplitude=10pt,mirror,raise=4pt},yshift=0pt] 
                (9,-3) -- (9,-1) node [black,midway,xshift=1.6cm] (Hd1) {$\hat{H}_{1, k+1}(\e_{M^{\delta}}^{\delta})$};
                
                \node[right of = mk2_1, xshift = -0.6cm, yshift = -0.2cm] (tmp1) {};
                \node[right of = mk2_2, xshift = -0.6cm, yshift = 0.2cm] (tmp2) {};
                \draw (tmp1) -- (tmp2) [dashed] node [midway, above, sloped] () {\tiny{$1, \dots, M^{\xi}$}};
                \draw [decorate,decoration={brace,amplitude=10pt,mirror,raise=4pt},yshift=0pt] 
                (9,0) -- (9,2) node [black,midway,xshift=1.4cm] (Hd1) {$\hat{H}_{1, k+1}(\e_1^{\delta})$};
                
                \draw [decorate,decoration={brace,amplitude=10pt,mirror,raise=4pt},yshift=0pt] 
                (11.5,-2.7) -- (11.5,1.7) node [black,midway,xshift=1.6cm] {$\widehat{E}_{{k, d}}[\hat{H}_{1, k+1}]$};
                
            \end{tikzpicture}
        \caption{Illustration of how $\widehat{E}_{{k, d}}[\hat{H}_{1, k+1}]$ is obtained using UT 
        over epistemic uncertainties. Here $\hat{H}_{1, k+1}(\e_m^{\delta})$ for $m = 1, \dots, M^{\xi}$ is obtained from the MCIS estimates of $\alpha(\hat{\xi}_{k+1}(\e^{\xi}_j))$.}
        \label{fig:schematic_approx}
    \end{figure*}
\end{remark}

\subsection{Stopping criterion}
\label{sec:myopic_conv}
For design strategies that make use of heuristic acquisition functions, it can be challenging to 
determine an appropriate stopping criterion. Here, we have considered the approximation $\hat{H}_{1, k}$ which 
has a natural interpretation. Hence, even if we make use of a criteria such as $\hat{H}_{2, k}$ or $\hat{H}_{3, k}$
to determine the next optimal decision, it makes sense to use $\hat{H}_{1, k}$ as an indicator of when 
the potential uncertainty reduction from future experiments is diminishing. 

We will let $\widehat{E}[\hat{\alpha}_k]$ and $\hat{H}_{1, k}$ be given as in \eqnref{eq:H1k_approx_2}, and define
\begin{equation}
    \label{eq:pof_cov}
    \hat{V}_k = \frac{\sqrt{\hat{H}_{1, k}}}{\widehat{E}[\hat{\alpha}_k]}.
\end{equation}
Then $\hat{V}_k$ is the UT-MCIS approximation of the \emph{coefficient of variation} of the failure probability $\alpha_k$ with respect to epistemic uncertainty. We will let $\hat{V}_k \leq V_{\text{max}}$ for some 
threshold $V_{\text{max}}$ serve as a criterion for stopping the myopic iteration procedure, in the case 
where a predefined maximum number of iterations $K_{\text{max}}$ has not already been reached. 

\begin{remark}
    The coefficient of variation is often used as a numerical criterion for convergence in Monte Carlo simulation.
    In structural reliability analysis, a coefficient of variation below $0.05$ is often used as an acceptable 
    level for failure probability estimation. 
    
    Note also that the criterion $\hat{V}_k \leq V_{\text{max}}$ for arbitrary $V_{\text{max}} \geq 0$ implicitly
    assumes that the epistemic uncertainty can be reduced to zero in the limit. If this is not the case, one might instead
    consider stopping when $\hat{V}_k$ is no longer decreasing. A different stopping criterion is also considered 
    in \secref{sec:example_4}.
\end{remark}

\subsection{Algorithm}
\label{sec:myopic_alg}

The complete procedure for myopic optimization is summarized in \algref{alg:myopic}.
Note that for simplicity the number of MCIS samples $N_0$ and $n$ are specified as input, 
but one may also consider deciding these using \eqnref{eq:split_MC} and \eqnref{eq:split_MC_var}. 
Using a standard technique in Monte Carlo simulation, one could keep increasing $N_0$ and $n$ 
until the coefficient of variation ($\text{std} / \text{mean}$) of the relevant estimator is sufficiently small.

\begin{algorithm}
    \DontPrintSemicolon
    \SetKwInOut{Input}{input}
    
    \Input{Model and sigma-points: $(\hat{\xi}_0, \hat{\delta}_0)$ and $\{ (v^{\xi}_{j}, \e^{\xi}_{j}) \}$, $\{ (v^{\delta}_{j}, \e^{\delta}_{j}) \}$. \newline 
    Number of samples for UT-MCIS and threshold: $N_0, n \in \mathbb{N}$ and $\tau > 0$, \newline 
    Max number of iterations and convergence criteria: $K_{\text{max}}$ and $V_{\text{max}}$.}

    \For{$k=0$ \KwTo $K_{\text{max}} - 1$}{
        (1) Generate samples $\{ (\x_i, w_i, \hat{\eta}_i) \}$ as described in \secref{sec:sampling} and compute 
        $\bar{h}_1 = \frac{1}{N_0} \sum_{|\hat{\eta}_i| \geq \tau} \ind{\hat{\eta}_i \leq 0} w_i.$\;
        (2) Compute $\hat{V}_k$ as in \eqref{eq:pof_cov} \;
        \eIf{$\hat{V}_k > V_{\text{max}}$}{
            (3) Compute the set $\{ \hat{\xi}^{m, i, j}_{k+1} \}$ as in \eqnref{eq:xi_k1_samples} and 
            define the function $\hat{J}_{i, k}(d)$ as in \eqnref{eq:approx_acq} \newline
            (for i = 1, 2, or 3 depending on the acquisition function of choice)\;
            (4) Find the optimal decision: $d_k = \argmin_{d \in \mathbb{D}} \hat{J}_{i, k}(d)$ \;
            (5) Make decision $d_k$ and obtain $(d_k, o_k)$\;
            (6) Update model $(\hat{\xi}_{k+1}, \hat{\delta}_{k+1}) = (\hat{\xi}_{k}, \hat{\delta}_{k}) | (d_k, o_k)$ \;
        }{
            Break. Convergence has been reached before $K_{\text{max}}$ iterations.
        }
    }
\caption{Myopic optimization}
\label{alg:myopic}
\end{algorithm}

\section{Numerical experiments}
\label{sec:num_exp}
Here we present a few numerical experiments using the algorithm for myopic 
optimal design presented in \secref{sec:myopic_alg}. 
Four experiments are presented, each with it's own objective: 
\begin{itemize}
    \item[1)] \secref{sec:example_1}: A toy example in 1d for conceptual illustration of the sequential design procedure.
    \item[2)] \secref{sec:example_2}: A hierarchical model with multiple 'expensive' sub-components.
    \item[3)] \secref{sec:example_3}: A non-hierarchical benchmark problem for comparison against 
    alternative strategies.
    \item[4)] \secref{sec:example_4} - A model that is more in resemblance of a realistic 
    application in structural reliability analysis, where we introduce different \emph{types} of decisions
    by considering both probabilistic function approximation and Bayesian inference of model parameters 
    through measurements with noise. 
\end{itemize}
All numerical experiments have been performed using \algref{alg:myopic} with the parameters 
$\tau = 3$, $N_0 = 10^4$, $n = 10^3$ and $V_{\text{max}} = 0.05$.
The probabilistic surrogate models used in the examples are all Gaussian process (GP) models 
with Mat\'ern $5/2$ covariance. 
A short summary of the relevant Gaussian process theory is given in \appref{app:GP}.

\subsection{Example 1: Illustrative 1d example}
\label{sec:example_1}
To illustrate the myopic procedure, we present a simple 1d example similar to the one 
given in \citep{Bect:2012:Sequential_design}, where we aim to emulate the limit state function 
\begin{equation}
\begin{split}
    g(x) &= 1 - \bigl( (0.4x - 0.3)^2 + \exp(-11.534|x|^{1.95})\\
    &+ \exp(-5(x-0.8)^2) \bigr).
\end{split}
\end{equation}
We assume that $g(x)$ can be evaluated at any $x \in \mathbb{R}$ without error, but that 
function evaluations are expensive. We will let $\xi(x)$ be the probabilistic surrogate in the 
form of a Gaussian process, where we use a prior mean $\mu(x) = -0.5$ together with a Mat\'ern $5/2$ covariance
function with fixed kernel variance $\sigma_c^2 = 0.1$ and length scale $l = 0.5$. 

\begin{figure*}[ht]
    \centering
    \includegraphics[width=0.75\linewidth, trim={0cm 0cm 0cm 0cm},clip]{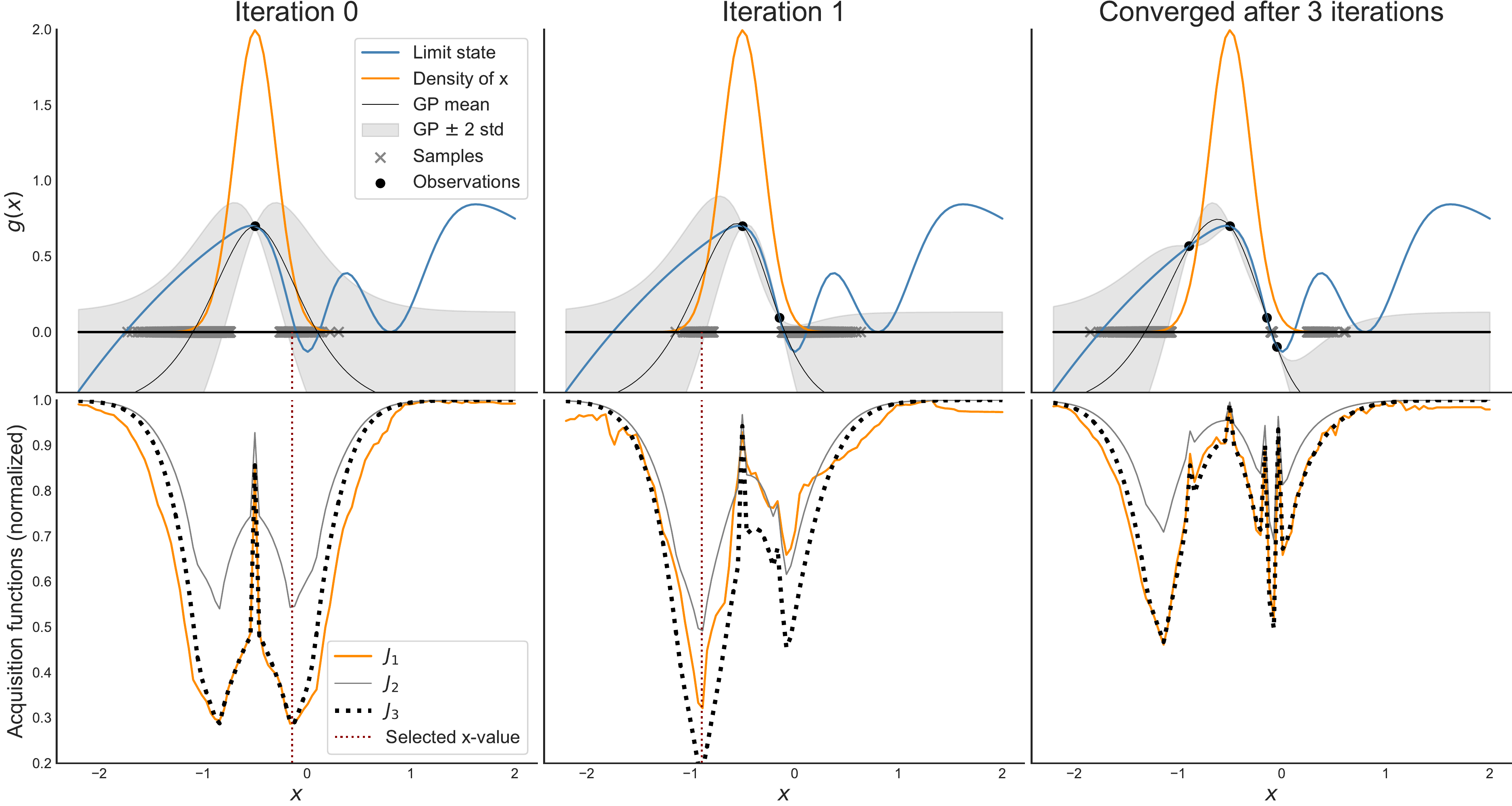}
    \caption{
        (Example 1) The top row shows the true limit state function $g(x)$, the probability density of $X$, 
        and the mean $\pm \ 2$ standard deviations of the GP $\xi_k$ for $k = 0, 1$ and $3$. The samples indicated 
        with {\color{gray}$\times$} on the $x$-axis are used in the importance sampling estimates of $J_{1, k}, J_{2, k}$ and $J_{3, k}$
        that are shown in the bottom row.
    }
    \label{fig:example_1_1}
\end{figure*}

We assume that $X$ follows Normal distribution with mean $\mu_X = -0.5$ and standard deviation $\sigma_X = 0.2$, 
and our goal is to estimate $\alpha(g) = P(g(X) \leq 0)$ using only a small number of 
evaluations of $g(\cdot)$. The set of decisions is therefore $\mathbb{D} = \cup_x \{ \text{evaluate } g(x) \}$ with 
respective outcomes $o(x) = g(x)$, and a predictive model for outcomes given as $\delta(x) = \xi(x)$.

Using a large number of samples of $g(X)$ we estimate $\alpha(g) \approx 0.0234$, 
and we will consider this as the 'true' failure probability for comparison.

We initiate $\xi$ by evaluating $g(x)$ at $x = \mu_X$. For subsequent 
function evaluations, we minimize the expected variance in the failure probability. I.e. 
we minimize the acquisition function $J_{1, k}$ given in \eqnref{eq:acquisition} with $\lambda \equiv 1$.
For comparison we also evaluate $J_{2, k}$ and $J_{3, k}$, and in this example it seems that 
all three acquisition functions would perform equally well. 
\figref{fig:example_1_1} shows $\xi_k$ and the corresponding three acquisition functions 
for the first few experiments, and \figref{fig:example_1_2} shows how $\alpha(\xi_k)$ evolves 
before converging after $k = 3$ iterations. 

\begin{figure}
    \centering
    \includegraphics[width=1.0\linewidth, trim={0.4cm 2.5cm 0.4cm 3.5cm},clip]{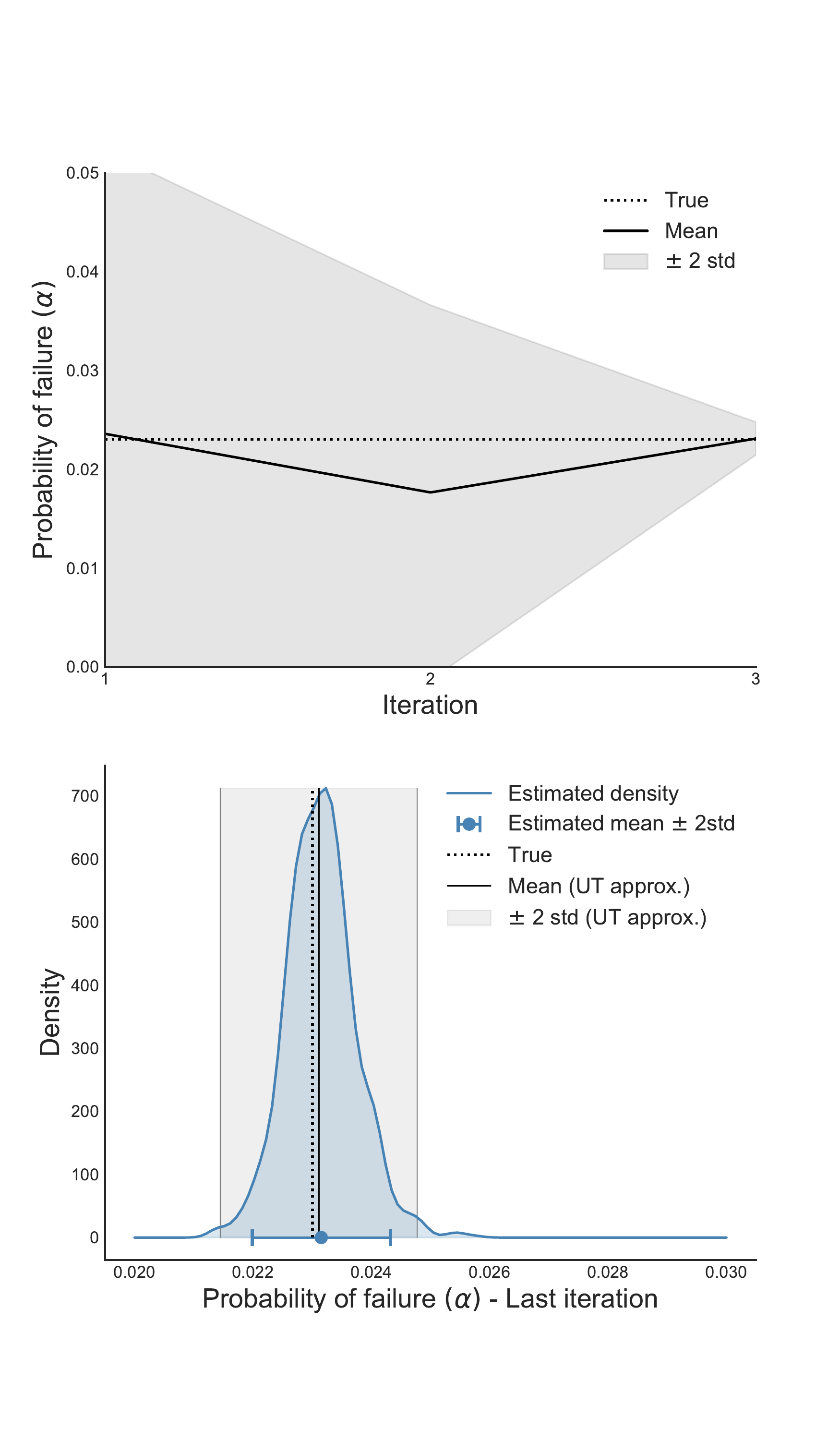}
    \caption{
        (Example 1) Top: Mean $\pm \ 2$ standard deviations of $\alpha_k$ after $k$ iterations, as computed 
        using the approximation described in \secref{sec:approx_H}.
        Bottom: The distribution of $\alpha_k$ at the final iteration $k = 3$, 
        estimated from a double-loop Monte Carlo (i.e. by sampling from $\alpha(\xi_k)$ without any approximation).
    }
    \label{fig:example_1_2}
\end{figure}

\subsection{Example 2: A 3 layer hierarchical model with 7D input}
\label{sec:example_2}
In this example we consider the structural reliability benchmark problem given as 
problem RP38 in \citep{Rozsas:2019:SRA_repo}. 
Here, $\x = (x_1, \dots, x_7) \in \mathbb{X} = \mathbb{R}^7$, 
and the limit state function $g(\x)$ can be written in terms of intermediate variables 
as follows:
\begin{equation}
    \begin{split}    
        &y_1(\x) = \frac{x_1 x_2^3}{2c_4 x_3^3}, \ \ 
        y_2(\x) = \frac{x_4^2}{c_2},\\
        &y_3(\x) = -4 x_5 x_6 x_7^2 + x_4 (x_6 + 4x_5 + 2x_6 x_7), \\
        &y_4(\x) = x_4 x_5 (x_4 + x_6 + 2x_6 x_7),
    \end{split}
\end{equation}
\begin{equation}
    z_1(\y) = \frac{c_4 y_1 y_2}{c_3},
\end{equation}
\begin{equation}
    g(\y, z_1) = 1 - \frac{c_2 c_3 z_1 + c_4 y_1 y_3}{c_1 y_4},
\end{equation}
where $c_1, c_2, c_3$ and $c_4$ are constants: 
\begin{equation}
\begin{split}
    &c_1 = 15.59 \cdot 10^4, \ \ 
    c_2 = 6 \cdot 10^4,\ \
    c_3 = 2 \cdot 10^5, \\ 
    &c_4 = 1 \cdot 10^6. 
\end{split}
\end{equation}

\figref{fig:3l_7d_problem} shows a graphical representation of how $g(\x)$ depends on the intermediate 
variables $z_1, y_1, y_3$ and $y_4$. We will assume that the functions $y_2(\x)$ and $z_1(\y)$ will require 
probabilistic surrogates, where $y_2(\x)$ and $z_1(\y)$ can be evaluated without error for any input $\x$ and $\y$.
We will also assume that there is no difference in the cost associated with evaluating $y_2$ or $z_1$, 
and our goal is to estimate the failure probability $\alpha(g)$ while keeping the total number of function evaluations of 
$y_2(\x)$ and $z_1(\y)$ as small as possible. Note that the effective domain of $y_2$ is $1$-dimensional 
and the effective domain of $z_1$ is $2$-dimensional. Hence, using surrogates for $y_2$ and $z_1$ 
should be much more efficient than building a single surrogate for $g$ using samples $g(\x_i)$. 
\begin{figure}[H]
    \centering
        \begin{tikzpicture}
            \tikzstyle{roundnode} = [circle, minimum size=0.8cm, text centered, draw=black]
            \tikzstyle{roundnode_fill} = [circle, minimum size=0.8cm, text centered, draw=black, fill = black!15]
            \tikzstyle{arrow} = [thick,->,>=stealth]

            \node (x) [roundnode] {$\x$};
            \node (y2) [roundnode_fill, below of=x, yshift=-0.5cm, xshift=-0.5cm] {$y_2$};
            \node (y1) [roundnode, left of=y2, xshift=-0.5cm] {$y_1$};
            \node (y3) [roundnode, right of=y2, xshift=0.5cm] {$y_3$};
            \node (y4) [roundnode, right of=y3, xshift=0.5cm] {$y_4$};
            \node (z1) [roundnode_fill, below of=y2, yshift=-0.5cm] {$z_1$};
            \node (g) [roundnode, below of=z1, yshift=-0.5cm] {$g$};

            \draw [arrow] (x) -- (y1);
            \draw [arrow] (x) -- (y2);
            \draw [arrow] (x) -- (y3);
            \draw [arrow] (x) -- (y4);
            \draw [arrow] (y1) -- (z1);
            \draw [arrow] (y2) -- (z1);
            \draw [arrow] (y1) -- (g);
            \draw [arrow] (z1) -- (g);
            \draw [arrow] (y3) -- (g);
            \draw [arrow] (y4) -- (g);
            
        \end{tikzpicture}
    \caption{(Example 2) Hierarchical representation of $g(\x)$. We assume that the intermediate variables $y_2(\x)$ and $z_1(\y)$ are expensive to evaluate.}
    \label{fig:3l_7d_problem}
\end{figure}
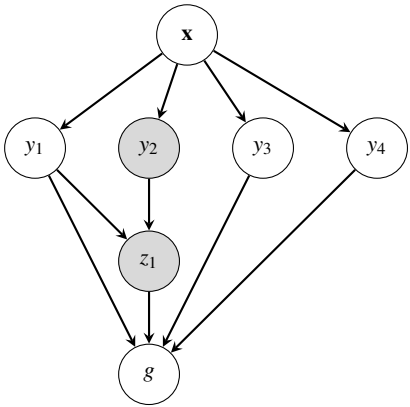

\begin{figure*}[ht]
    \centering
    \includegraphics[width=0.75\linewidth, trim={0cm 0cm 0cm 0cm},clip]{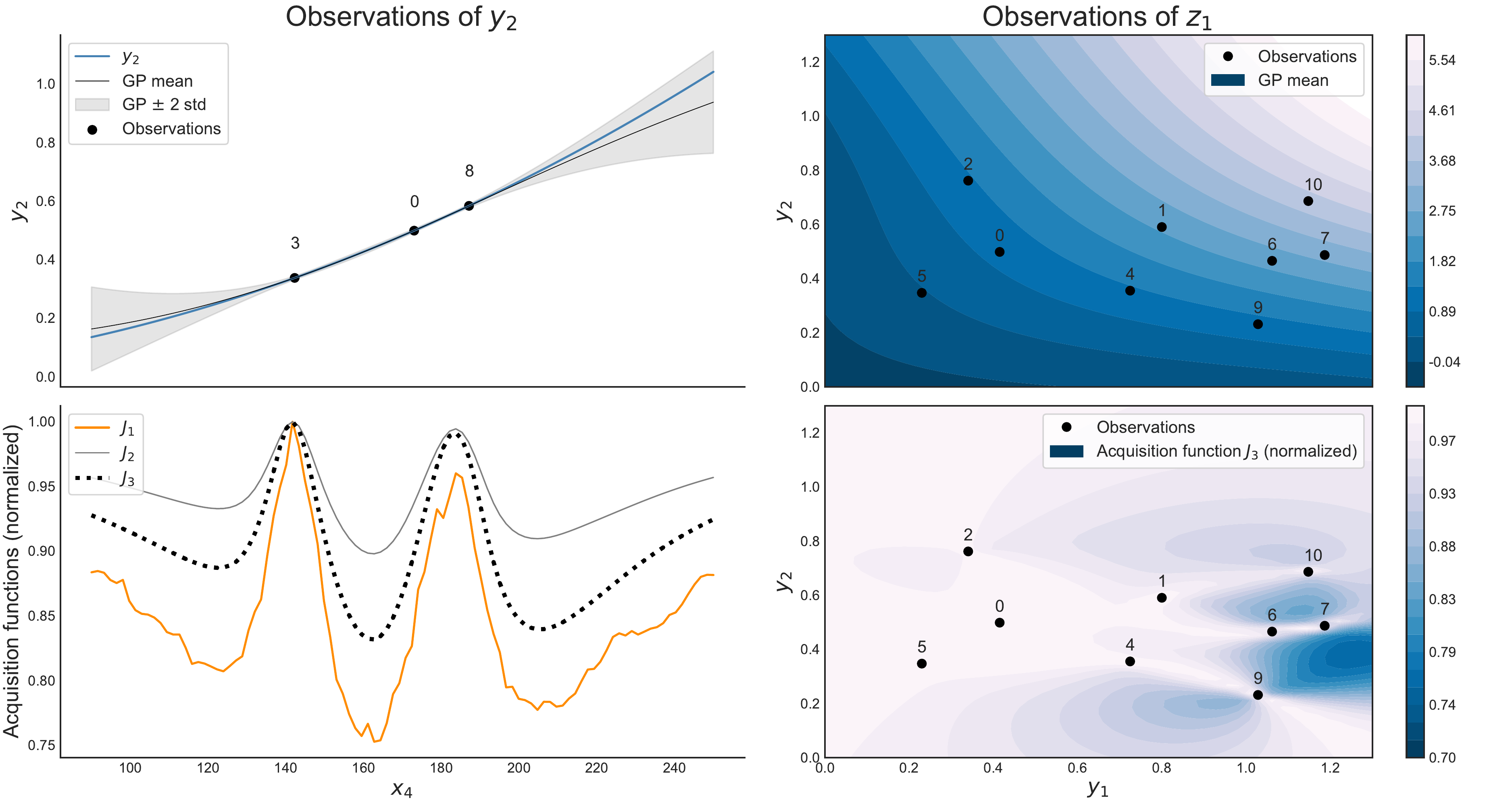}
    \caption{
        (Example 2) The top row shows the GP models $\tilde{y}_2$ and $\tilde{z}_1$ with respect to $P_k$ for $k = 10$. 
        The number above each observations is the iteration index $k$,
        and convergence was obtained after $2$ evaluations of $y_2$ and $8$ evaluations of $z_1$.
        The final acquisition functions are shown in the bottom row.
    }
    \label{fig:example_2_1}
\end{figure*}

As for the random variable $\X = (X_1, \dots, X_7)$, we assume that all $X_i$'s are independent and normally distributed, 
$X_i \sim \mathcal{N}(\mu_i, \sigma_i)$, with means 
$\mu_1 = 350$, 
$\mu_2 = 50.8$,
$\mu_3 = 3.81$,
$\mu_4 = 173$,
$\mu_5 = 9.38$,
$\mu_6 = 33.1$,
$\mu_7 = 0.036$,
and standard deviation $\sigma_i = 0.1\mu_i$. The 'true' failure probability we aim to estimate 
is $\alpha(g) \approx 8.1 \cdot 10^{-3}$.

Assuming $y_2$ and $z_1$ are expensive to evaluate, we introduce two Mat\'ern $5/2$ GP surrogates, 
$\tilde{y}_2$ and $\tilde{z}_1$. The initial kernel parameters are $(\sigma_c^2 = 0.03, l = 20)$ and 
$(\sigma_c^2 = 2, l = [0.5, 0.5])$ for $\tilde{y}_2$ and $\tilde{z}_1$ respectively. These parameters may
be updated by maximum likelihood estimation, but not until a few observations (resp. 2 and 5) have been made.
We know that large values of $y_2$ or $z_1$ will result in poor structural performance (small $g(\x)$), 
so we initiate the GP models with conservative prior means of $\mu(\x) = 1$ for $\tilde{y}_2$ and $\mu(\y) = 5$ for $\tilde{z}_1$.
Both models are initially updated with one observation each, $\tilde{y}_2(\mu_4) = y_2^0$
and $\tilde{z}_1(y_1^0, y_2^0) = z_1^0$ for $y_1^0 = y_1(\mu_1, \mu_2, \mu_3)$, $y_2^0 = y_2(\mu_4)$ and $z_1^0 = z_1(y_1^0, y_2^0)$.

In this example, we would then define $\xi(\x) = g(y_1, \tilde{z}_1, \tilde{y}_2, y_3, y_4)$. 
With respect to $\tilde{z}_1$, there is a set of possible decision for uncertainty reduction, 
namely $\mathbb{D} = \cup_{y_1, y_2} \{ \text{evaluate } z_1(y_1, y_2) \}$, 
with a corresponding set of observations $\mathbb{O} = \cup_{y_1, y_2} \{ z_1(y_1, y_2) \}$, 
and a predictive model $\delta(y_1, y_2) = \tilde{z}_1(y_1, y_2)$. Similarly, we obtain a set 
of decisions, outcomes and a predictive model for $\tilde{y}_2$, and we can update 
$\mathbb{D}, \mathbb{O}$ and $\delta(d)$ accordingly.

Convergence was reached at iteration $k = 10$, after $2$ additional evaluations of $y_2$ and $8$ additional evaluations of $z_1$.
\figref{fig:example_2_1} shows the updated surrogate models, $\tilde{y}_2 | I_k$ and $\tilde{z}_1 | I_k$ for $k = 10$, 
and \figref{fig:example_2_2} shows how $\alpha(\xi_k)$ evolves with each iteration. 
At each iteration, the next experiment was decided by minimizing the acquisition function $J_{3, k}$ 
with respect to updating each of the two surrogate models.  

\begin{figure}[H]
    \centering
    \includegraphics[width=1.0\linewidth, trim={0.4cm 2.5cm 0.4cm 3.5cm},clip]{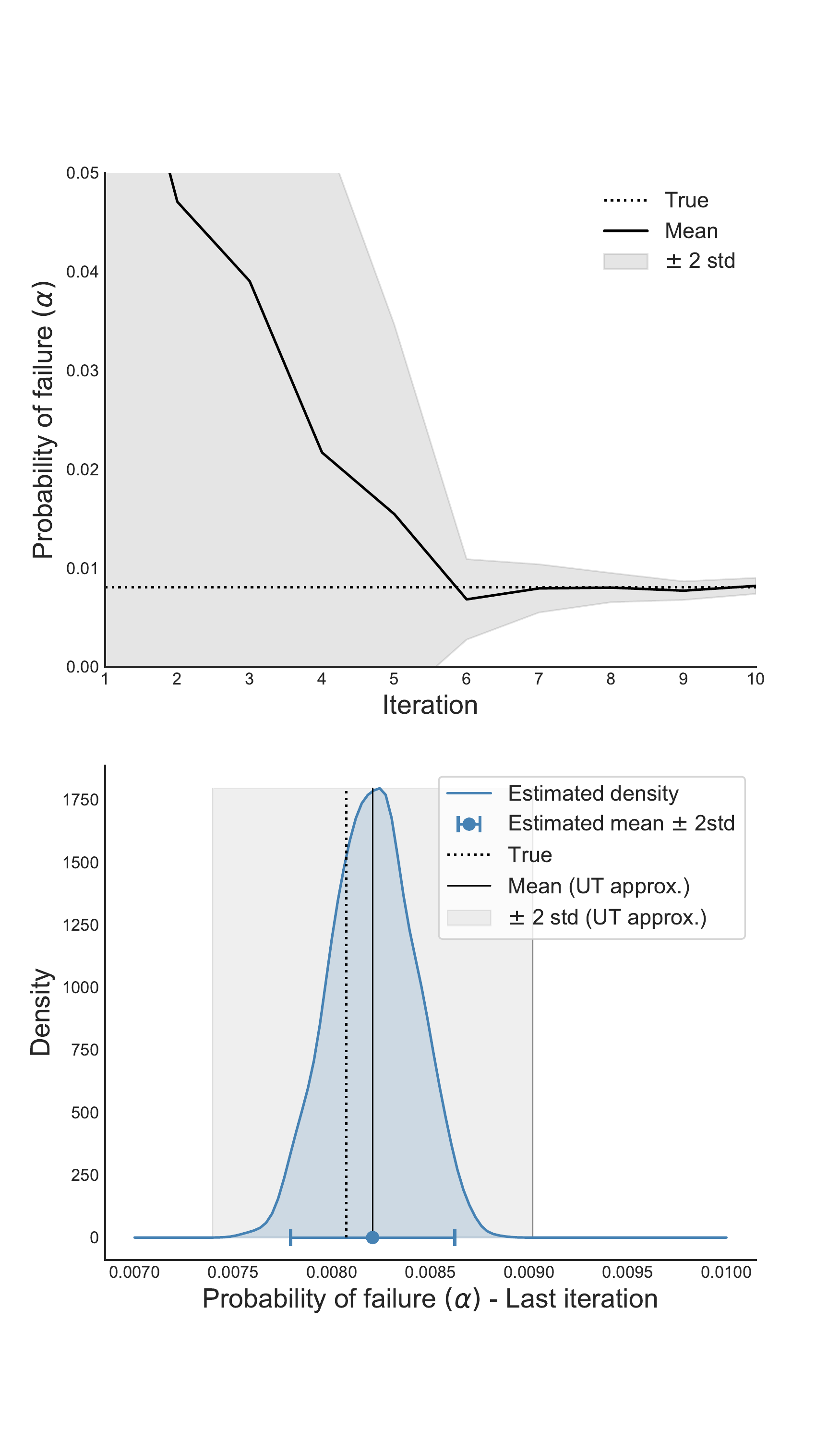}
    \caption{
        (Example 2) Top: Mean $\pm \ 2$ standard deviations of $\alpha_k$ after $k$ iterations, as computed 
        using the approximation described in \secref{sec:approx_H}.
        Bottom: The distribution of $\alpha_k$ at the final iteration $k = 10$, 
        estimated from a double-loop Monte Carlo (i.e. by sampling from $\alpha(\xi_k)$ without any approximation).
    }
    \label{fig:example_2_2}
\end{figure}

\subsection{Example 3: The 4 branch system}
\label{sec:example_3}
Here we consider the 'four branch system', a classical 2D benchmark problem given by the limit state
\begin{equation}
    \label{eq:four_branch}
    g(\x) = \min 
    \left\{
    \begin{array}{l}
      3 + 0.1(x_1 - x_2)^2 - (x_1+ x_2) \/ \sqrt{2}; \\
      3 + 0.1(x_1 - x_2)^2 + (x_1+ x_2) \/ \sqrt{2}; \\
      (x_1 - x_2) + 6 \/ \sqrt{2}; \\
      (x_2 - x_1) + 6 \/ \sqrt{2}
    \end{array}
  \right\},
\end{equation}
and where $x_1$ and $x_2$ are independent standard normal variables. 
In this example we will not write \eqnref{eq:four_branch} as an hierarchical model, 
in order to compare our method with other alternatives that are tailored to to non-hierarchical setting. 
We therefore let $\xi(\x)$ be a Gaussian process surrogate of $g(\x)$, constructed from observations $(\x_i, g(\x_i))$.
For the initial 'conservative' Gaussian process we select a prior mean of $-1$, a Mat\'ern $5/2$ kernel with parameters 
of $(\sigma_c = 1, l = 3)$, and condition on the initial observation $(\bm{0}, g(\bm{0}))$.

According to \cite{Huang:2017:SRA}, the method called AK-MCS developed by \cite{Echard:2011:AKMCS}
is considered a typical and mature approach, and should therefore be a suitable candidate for comparison. 
In addition, \cite{Echard:2011:AKMCS} also provide the results from using a number of other alternatives proposed in
\cite{Schueremans:2005:splines_and_NN}. \tblref{tbl:4branch_res} gives a summary of the results from \cite{Echard:2011:AKMCS},
together with the those obtained using the approach presented in this paper. 
\begin{table}
    \footnotesize
    \centering
    \caption{(Example 3) Table 2 from \cite{Echard:2011:AKMCS}, where we have appended the method 
    from this paper (UT-MCIS) in the bottom row. The reported failure probabilities ($\widehat{p}_f$)
    are the estimated mean $\pm$ 2 standard deviations of $\alpha(\xi_k)$ for $k = 35$ 
    (stopped at $\hat{V}_k \leq 0.1$), $k = 48$ (stopped at $\hat{V}_k \leq 0.05$), and
    $k = 65$ (stopped at $\hat{V}_k \leq 0.025$).
    } 
    \begin{tabular}{llll}
        \toprule
        Method &  $N_{\text{call}}$ & $\widehat{p}_f \times 10^{3}$  \\
        \midrule
        Monte Carlo  & $10^6$ & $4.416$ \\
        AK-MCS+U     & 126    & $4.416$    \\
        AK-MCS+EFF   & 124    & $4.412$   \\
        \\
        Directional Sampling (DS)  & 52   & $4.5$   \\
        DS + Response Surface      & 1745 & $5.0$   \\
        DS + Spline                & 145  & $2.4$   \\
        DS + Neural Network        & 165  & $4.1$   \\
        \\
        Importance Sampling (IS)  & 1469 & $4.9$   \\
        DS + Response Surface     & 1375 & $4.5$   \\
        IS + Spline               & 428  & $4.5$   \\
        IS + Neural Network       & 52   & $5.7$   \\
        \midrule
        UT-MCIS  ($V_{max} = 2.5 \%$)  & 65   & $(4.347- 4.444)$   \\
        UT-MCIS  ($V_{max} = 5 \%$)    & 48   & $(4.288- 4.470)$  \\
        UT-MCIS  ($V_{max} = 10 \%$)   & 35   & $(4.163- 4.547)$    \\
        \bottomrule
    \end{tabular}      
    \label{tbl:4branch_res}
\end{table}
Our results in \tblref{tbl:4branch_res} are obtained using \algref{alg:myopic} with three different stopping criteria, 
$V_{max} = 0.1$, $V_{max} = 0.05$ and $V_{max} = 0.025$. 
Instead of point estimates we provide prediction intervals, which in this example contain the 
'true' failure probability obtained with Monte Carlo in each scenario.
From a practical perspective, even the estimates obtained using only $35$ evaluations ($V_{max} = 0.1$) of
\eqnref{eq:four_branch} seems acceptable. 
If we were to use the mean + 2 standard deviations as a conservative estimate, 
the relative error with respect to the 'true' failure probability is still less than 3 \%.
After an additional $30$ iterations, this number drops to 0.65 \%. 
Hence, our approach performs well with respect to the alternatives considered in \citep{Echard:2011:AKMCS, Schueremans:2005:splines_and_NN}. It should also be noted that the Directional Sampling alternative in \tblref{tbl:4branch_res}
is a method that is especially suitable for the specific 'radial' type of limit state surfaces as considered here, 
and a this level of performance is not expected in general. 

Optimization was performed using the approximate acquisition function $\hat{J}_{3, k}$, and \figref{fig:example_3_1} shows
how the sequence of observations are located with respect to the failure set $g = 0$. The resulting sequence of failure
probabilities after each iteration is illustrated in \figref{fig:example_3_2}.

\begin{figure}
    \centering
    \includegraphics[width=1.0\linewidth, trim={2cm 2cm 2cm 2cm},clip]{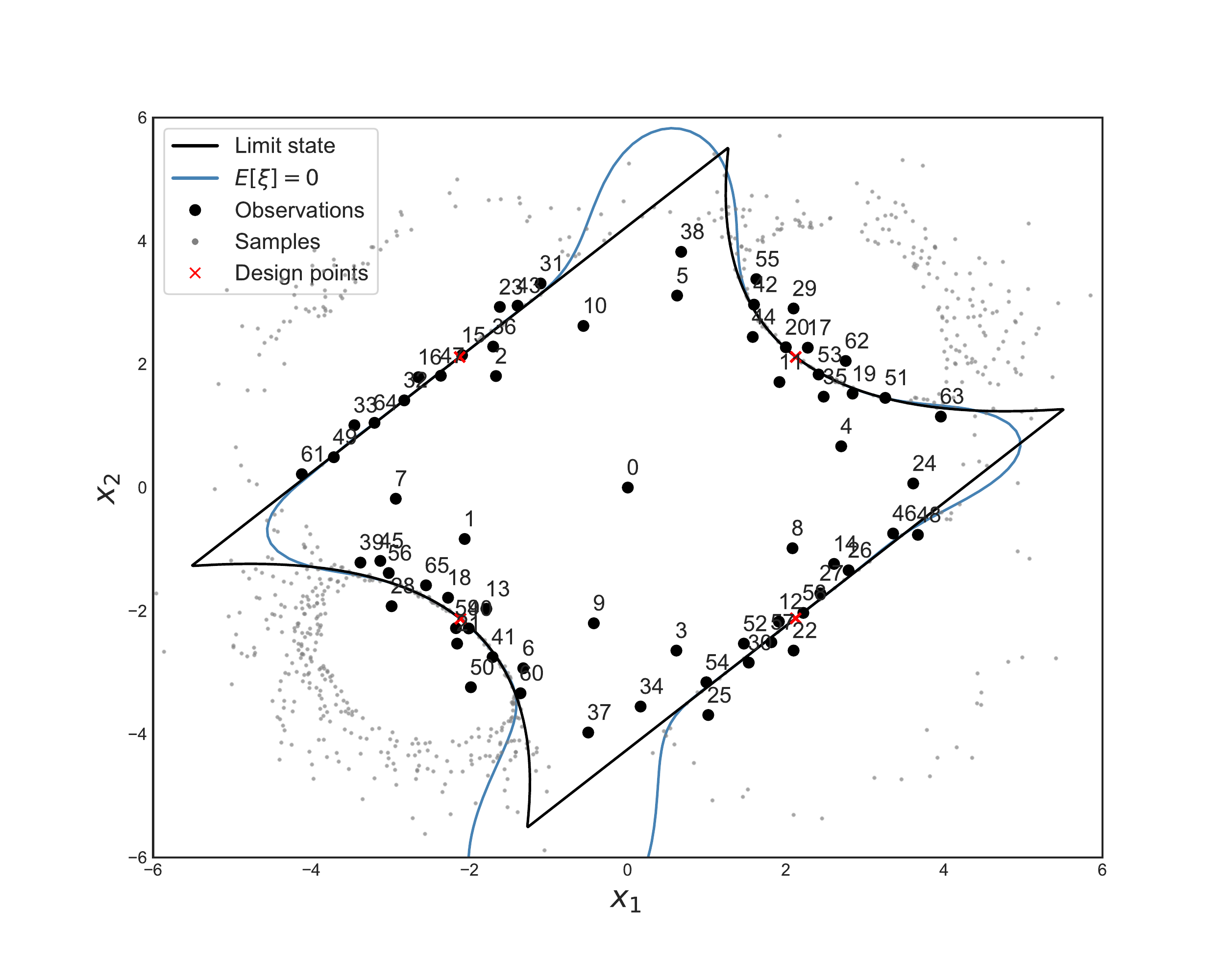}
    \caption{
        (Example 3) The limit state \eqnref{eq:four_branch} together with the expected failure surface $E[\xi_{65}]$ and the $65$ observations
        collected before convergence at $\hat{V}_{65} < 0.025$. 
        The proposal distribution $q_{\X}$ used for importance sampling is a mixture of Gaussian random variables centered 
        at the four design points ({\color{red}{$\times$}}) as described in \appref{app:sampling_dist_q}. 
        The pruned samples shown in the figure are mostly located around $E[\xi_{65}] = 0$ and in other regions where the level set $\xi_{65} = 0$ is uncertain.
    }
    \label{fig:example_3_1}
\end{figure}

\begin{figure}
    \centering
    \includegraphics[width=1.0\linewidth, trim={0.4cm 2.5cm 0.4cm 3.5cm},clip]{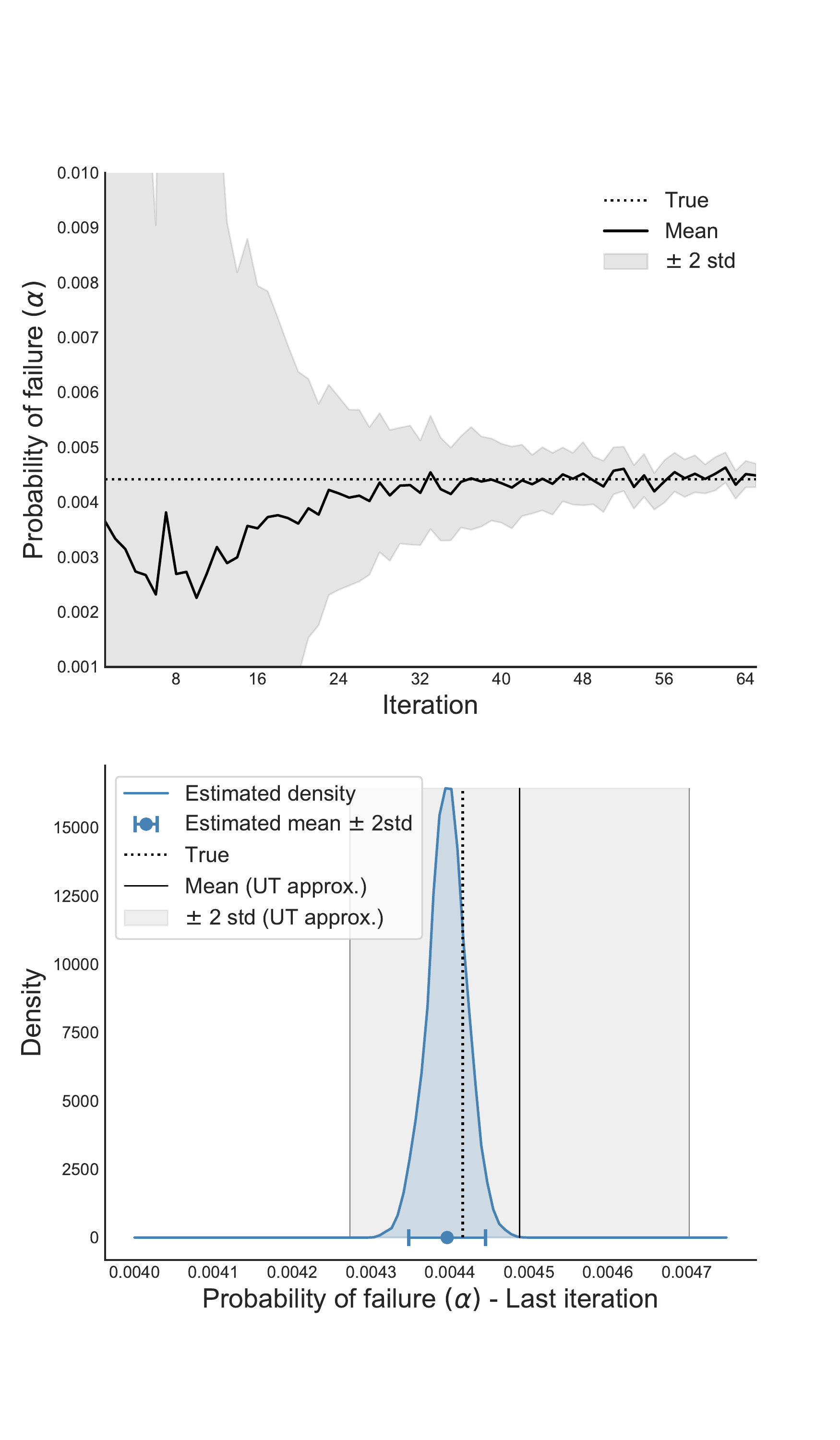}
    \caption{
        (Example 3) Top: Mean $\pm \ 2$ standard deviations of $\alpha_k$ after $k$ iterations, as computed 
        using the approximation described in \secref{sec:approx_H}.
        Bottom: The distribution of $\alpha_k$ at the final iteration $k = 65$, 
        estimated from a double-loop Monte Carlo (i.e. by sampling from $\alpha(\xi_k)$ without any approximation).
    }
    \label{fig:example_3_2}
\end{figure}

\subsection{Example 4: Corroded pipeline example}
\label{sec:example_4}
To give an example of a scenario where there are different \emph{types} of experiments,
we consider a probabilistic model which is recommended for engineering assessment of 
offshore pipelines with corrosion \citep{RP-F101}.
The failure mode under consideration is where a pipeline bursts, 
when the pipeline's ability to withstand the high internal pressure has been reduced as a consequence of corrosion.
\\

\noindent \textbf{The structural reliability model}\\
\noindent \figref{fig:corrpipe_problem} shows a graphical representation of the structural reliability model.
Here, a steel \textbf{pipeline} is characterised by the outer diameter ($D$ [mm]), the wall thickness ($t$ [mm]) 
and the ultimate tensile strength ($s$ [MPa]). In this example we let 
$D = 800$, $t \sim \mathcal{N}(\mu = 20, \text{cov} = 0.03)$, and 
$s \sim \mathcal{N}(\mu = 545, \text{cov} = 0.06)$, where cov is the coefficient of variation 
(standard deviation / mean).

The pipeline contains a rectangular shaped \textbf{defect} with a given depth ($d$ [mm]) and length ($l$ [mm]), where $l \sim \mathcal{N}(\mu = 200, \sigma^2 = 1.49)$ and where $d$ will 
be inferred from observations. 

Given a pipeline $(D, t, s)$ with a defect $(d, l)$, we can determine the pipeline's pressure 
resistance \textbf{capacity} (the maximum differential pressure the pipeline can withstand before bursting). 
We let $p_{\text{FE}}$ [MPa] denote the capacity coming from a Finite Element simulation of the physical phenomenon. 

From the theoretical capacity $p_{\text{FE}}$, we model the true pipeline capacity as 
$p_c = X_{\text{m}} \cdot p_{\text{FE}}$, where $X_{\text{m}}$ is the \textbf{model discrepancy}, $X_{\text{m}} \sim \mathcal{N}(\mu_{\text{m}}, \sigma_{\text{m}}^2)$.
For simplicity we have assumed that $X_{\text{m}}$ does not depend on the type of pipeline and defect, 
and we will also assume that $\sigma_{\text{m}} = 0.1$, where only the mean $\mu_{\text{m}}$ will be inferred from observations 
of the form $p_c / p_{\text{FE}}$.

Finally, the pressure \textbf{load} (in MPa) is modelled as a Gumbel distribution with mean $15.75$ and standard deviation $0.4725$. The limit state representing the transition to failure is then given as $g = p_c - p_d$.\\

\begin{figure}
    \centering
        \begin{tikzpicture}[scale=0.9, every node/.style={scale=0.9}]
            \tikzstyle{roundnode} = [circle, minimum size=0.9cm, text centered, draw=black]
            \tikzstyle{roundnode_fill} = [circle, minimum size=0.9cm, text centered, draw=black, fill = black!15]
            \tikzstyle{arrow} = [thick,->,>=stealth]

            \node (D) [roundnode] {$D$};
            \node (t) [roundnode, right of=D, xshift=0.5cm] {$t$};
            \node (UT) [roundnode, right of=t, xshift=0.5cm] {$s$};
            \node (dt) [roundnode_fill, right of=UT, xshift=0.75cm] {$d$};
            \node (l) [roundnode, right of=dt, xshift=0.5cm] {$l$};
            
            \node (mod_sig) [roundnode, below of=D, yshift=-1.0cm, xshift=-0.0cm] {$\sigma_{\text{m}}$};
            \node (mod_mu) [roundnode_fill, left of=mod_sig, xshift=-0.5cm] {$\mu_{\text{m}}$};
            \node (Xm) [roundnode, below of=mod_sig, yshift=-0.5cm] {$X_{\text{m}}$};
            
            \node (Pcap_FE) [roundnode_fill, below of=UT, yshift=-1.0cm] {$p_{\text{FE}}$};
            \node (Pcap) [roundnode, below of=Pcap_FE, yshift=-0.5cm] {$p_{c}$};
            \node (pd) [roundnode, right of=Pcap, xshift=1.5cm] {$p_{d}$};
            \node (g) [roundnode, below of=Pcap, yshift=-0.5cm, xshift = 1.5cm] {$g$};
            
            \draw [arrow] (D) -- (Pcap_FE);
            \draw [arrow] (t) -- (Pcap_FE);
            \draw [arrow] (UT) -- (Pcap_FE);
            \draw [arrow] (dt) -- (Pcap_FE);
            \draw [arrow] (l) -- (Pcap_FE);
            \draw [arrow] (mod_sig) -- (Xm);
            \draw [arrow] (mod_mu) -- (Xm);
            \draw [arrow] (Xm) -- (Pcap);
            \draw [arrow] (Pcap_FE) -- (Pcap);
            \draw [arrow] (Pcap) -- (g);
            \draw [arrow] (pd) -- (g);
            
            \node [above of = t, yshift = -0.25cm,] {\footnotesize{\textbf{Pipeline}}};
            \node [above of = dt, yshift = -0.25cm, xshift=0.8cm] {\footnotesize{\textbf{Defect}}};
            \node [above of = mod_mu, yshift = -0.25cm, xshift=0.8cm] {\footnotesize{\textbf{Model discrepancy}}};

            \node [above of = pd, yshift = -0.25cm, xshift = 0.0cm] {\footnotesize{\textbf{Load}}};
            \node [above of = Pcap, yshift = -0.25cm, xshift = -0.8cm, text width=1cm, align=center, rotate=90] {\footnotesize{\textbf{Capacity}}};

        \end{tikzpicture}
    \caption{(Example 4) Graphical representation of the corroded pipeline structural reliability model. 
    The shaded nodes $d$, $p_{\text{FE}}$ and $\mu_{\text{m}}$ have associated epistemic uncertainty 
    that can be reduced through experiments.}
    \label{fig:corrpipe_problem}
\end{figure}
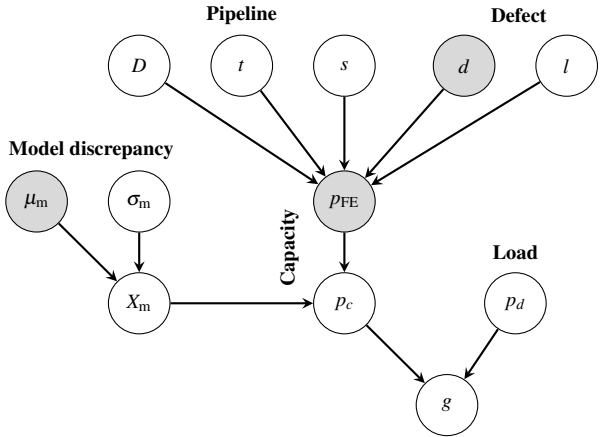

\noindent \textbf{Different types of decisions}\\
We consider the following three types of decisions
\begin{enumerate}
    \item \textbf{Defect measurement:} We assume that unbiased measurements of the relative depth $d/t$ can be obtained.
    The measurements come with additive Gaussian noise, 
    $\epsilon \sim \mathcal{N}(0, \sigma_{d/t}^2)$, and we will assume that three types of inspection 
    are available, corresponding to $\sigma_{d/t} = 0.02, 0.04$ and $0.08$.
    \item \textbf{Computer experiment:} Evaluate $p_{\text{FE}}$ at some deterministic input $(D, t, s, d, l)$.
    \item \textbf{Lab experiment:} Obtain one observation of $X_{\text{m}}$.
\end{enumerate}

In order to generate synthetic data for this experiment, we assume that the true defect depth is $d = 0.3t = 6$ mm
and that $\mu_m = 1.0$. Instead of running a full Finite Element simulation to obtain $p_{\text{FE}}$, we will
make use of the simplified capacity equation in \citep{RP-F101}, in which case
\begin{equation*}
    p_{\text{FE}} = 1.05\frac{2t s}{D - t} \frac{1 - d/t}{1 - \frac{d/t}{Q}} \text{,} \ \ \  Q = \sqrt{1 + 0.31\frac{l^{2}}{Dt}}.
\end{equation*}

\noindent \textbf{Results}\\
To define the initial model $\xi_0$ we need a prior specification over the epistemic quantities $d$, 
$\mu_{\text{m}}$ and $p_{\text{FE}}$.
We let $d$ be a priori normal with mean $0.5$ and standard deviation $0.15$, and $\mu_{\text{m}}$ normal 
with mean $1.0$ and standard deviation $0.1$. Consequently, the posteriors of $d$ and $\mu_{\text{m}}$ 
(and also $X_{\text{m}}$) given any number of observations are all normal. 
The function $p_{\text{FE}}$ is replaced by
a GP surrogate with prior mean $\mu = -10$ and $\sigma_c = 10$, $l = [1, 1, 1, 1]$
Mat\'ern $5/2$ parameters, which we initiate using a single observations at the expected value of the input. 

We assume that the computer experiments are cheap compared to the lab experiments, and that the direct measurements
of $d/t$ is most expensive. To reflect these varying costs, we specify the acquisition function
\begin{equation}
    \label{eq:approx_acq_ex4}
    \hat{J}_{i, k}(d) = c(d) \frac{\widehat{E}_{{k, d}}[\hat{H}_{i, k+1}]}{\hat{H}_{i, k}},
\end{equation}
where $c(d)$ is the cost of a given decision. (Note that in \eqnref{eq:approx_acq_ex4} the variable $d$
refers to a \emph{decision}, but for the remaining part of this example $d$ will only 
refer to the \emph{defect depth}). In \eqnref{eq:approx_acq_ex4} we have normalized the expected future 
measure of residual uncertainty with the current, which gives an estimate of the expected improvement 
given a certain decision. The numerical values representing difference in costs is given by 
$c = 1$ for computer experiments, $c = 1.1$ for lab experiments, and $c = 1.11, 1.12, 1.13$ 
for measurements of $d/t$ with accuracy $\sigma_{d/t} = 0.08, 0.04$ and $0.02$ respectively.

In structural reliability analysis, the objective is not always to obtain an estimate of the failure probability 
that is as accurate as possible. A relevant problem in practice is to determine whether a structure 
satisfies some prescribed target reliability level $\alpha_{target}$. 
In this example, we aim to either confirm that the failure probability is less than the 
target $\alpha_{target} = 10^{-3}$ (in which case we can continue operations as normal), 
or to detect with confidence that the target is exceeded (and we have to intervene). 
For this purpose we intend to stop the iterative procedure if the difference between the expected and target
failure probability is at least $4$ standard deviations. In addition to the standard stopping criterion for 
convergence \eqnref{eq:pof_cov}, we therefore introduce the stopping criterion 
\begin{equation}
    \label{eq:ex_4_stopping}
    \widehat{E}[\hat{\alpha}_k] + 4\sqrt{\hat{H}_{1, k}} < \alpha_{target}, \ \text{ or } \ 
    \widehat{E}[\hat{\alpha}_k] - 4\sqrt{\hat{H}_{1, k}} > \alpha_{target}.
\end{equation}
\figref{fig:example_4_1} shows how the UT-MCIS approximation of the failure probability evolves throughout $100$ iterations. 
We have made use of $i = 3$ in \eqnref{eq:approx_acq_ex4} as we found the corresponding acquisition surface for $p_{\text{FE}}$
smoother than the alternative $i = 1$, and hence easier to minimize numerically.
The stopping criterion \eqnref{eq:ex_4_stopping} is reached after $k = 25$ iterations, and 
\figref{fig:example_4_2} shows the corresponding posteriors of the relative defect depth $d / t$ and 
the model discrepancy $X_{\text{m}}$.

\begin{figure}
    \centering
    \includegraphics[width=1.0\linewidth, trim={0.4cm 2.5cm 0.4cm 3.5cm},clip]{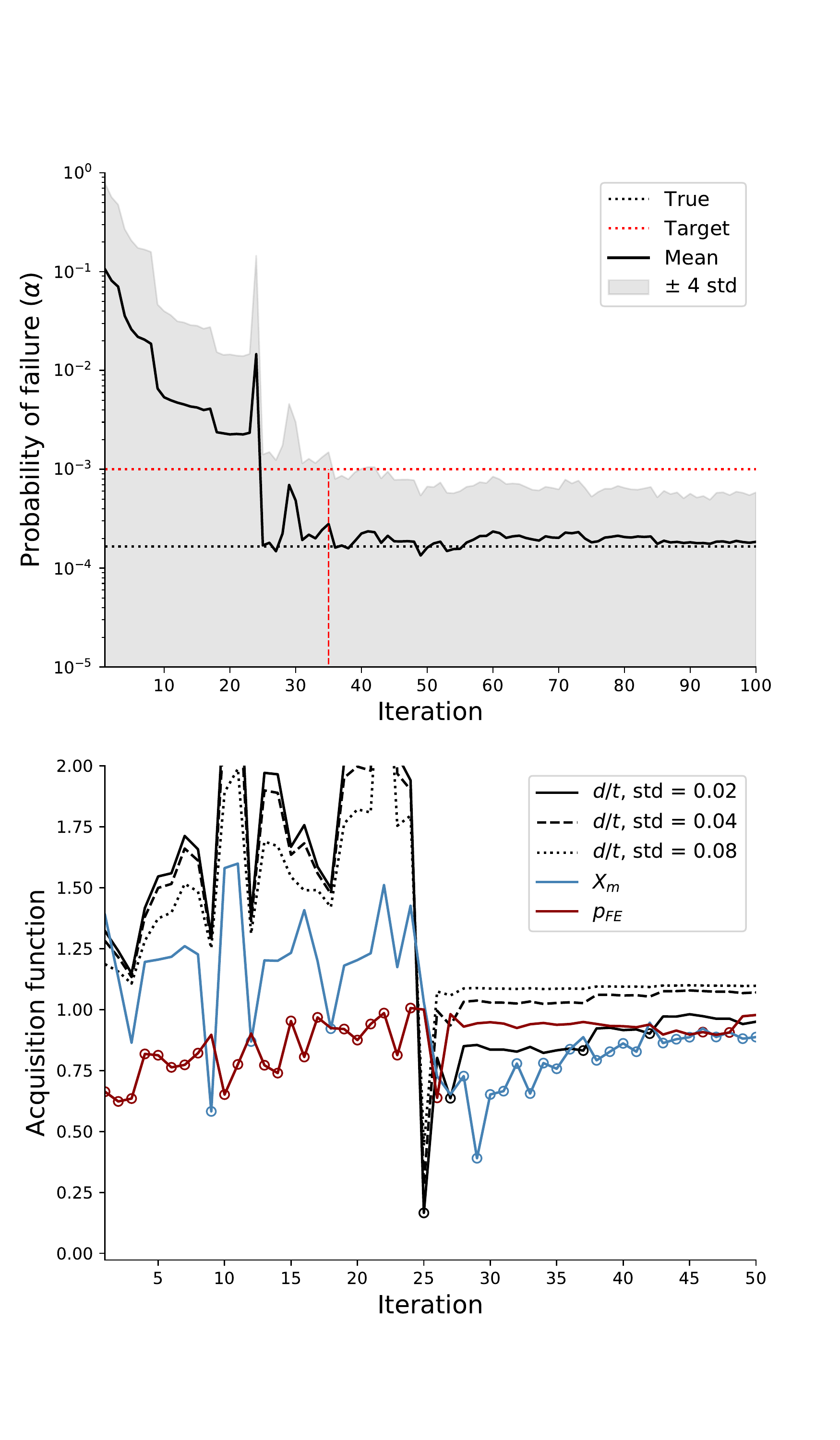}
    \caption{
        (Example 4) Top: Mean $\pm \ 4$ standard deviations of $\alpha_k$ after $k$ iterations, as computed 
        using the approximation described in \secref{sec:approx_H}. The stopping criterion \eqnref{eq:ex_4_stopping}
        is reached after $25$ iterations. 
        Bottom: The acquisition functions \eqnref{eq:approx_acq_ex4} for each type of experiment
        during the first $50$ iterations. 
    }
    \label{fig:example_4_1}
\end{figure}

\begin{figure}
    \centering
    \includegraphics[width=1.0\linewidth, trim={0.4cm 2.5cm 0.4cm 3.5cm},clip]{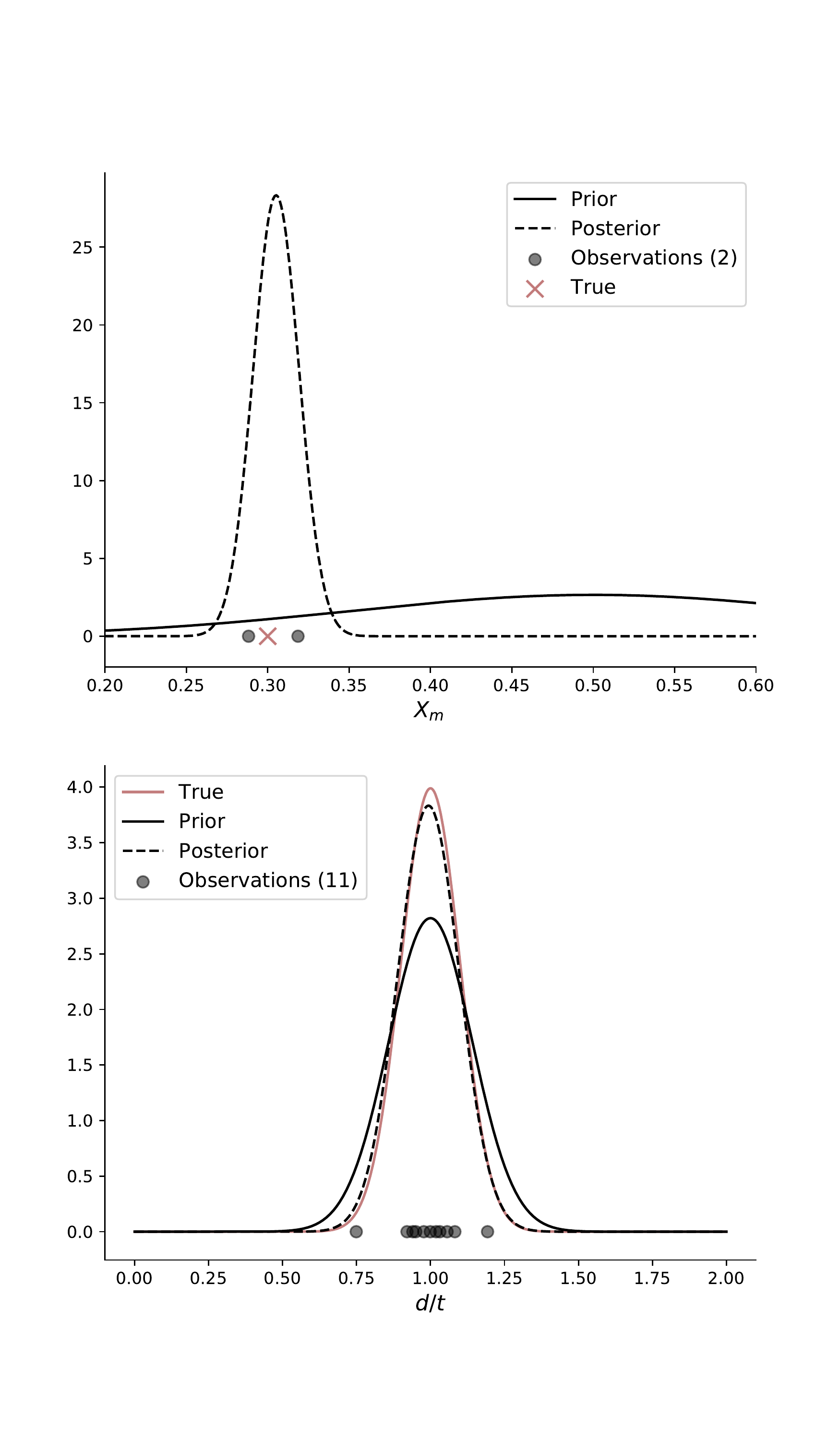}
    \caption{
        (Example 4) The posterior distributions of $d/t$ and $X_{\text{m}}$ when the stopping criterion \eqnref{eq:ex_4_stopping} 
        is reached at $k = 25$.
    }
    \label{fig:example_4_2}
\end{figure}

Throughout the examples in this paper we have initiated GP surrogate models using a single observation 
at the expected input. A different approach that is often found in practical applications is to 
initiate the GP surrogate with a space-filling design. A very common alternative is to 
make use of a Latin Hypercube sample (LHS), of size no more than $10 \ \times$ the input dimension
(although the appropriate number of samples naturally depends on how nonlinear the response is expected to be).  

\tblref{tbl:ex4_summary} shows a summary of the results from running this example with and without an
initial design consisting of $10$ LHS samples. For this example it does not seem to make any significant difference, 
but we see why the stopping criterion \eqnref{eq:ex_4_stopping} is useful, as 
on average we can conclude that the failure probability is below the target value after around $30$-$40$ iterations.
\begin{table}
    \footnotesize
    \centering
    \caption{(Example 4) Averages over $100$ runs, using $1$ versus $10$ initial observations of $p_{\text{FE}}$.} 
    \begin{tabular}{llcccc}
        \toprule
         Initial& Stop at & cov of $\alpha_k$ & \multicolumn{3}{c}{Number of observations}  \\
         design & target & $(\hat{V}_k)$ & $p_{\text{FE}}$ & $X_{\text{m}}$ & $d/t$ \\
        \midrule
        \multirow{2}{*}{$E[\X]$} & Yes & 1.39 & 23 + 1 & 10 & 2 \\
        & No & 0.63 & 46 + 1 & 47 & 7\\
        \midrule
        \multirow{2}{*}{LHS $10$} & Yes & 1.37 & 12 + 10 & 8 & 2\\
        & No & 0.90 & 45 + 10 & 48 & 7\\
        \bottomrule
    \end{tabular}      
    \label{tbl:ex4_summary}
\end{table}

We leave this numerical experiment with an important remark, which is that specifying an appropriate 
cost in \eqnref{eq:approx_acq_ex4} can be difficult. If for instance the cost related to a measurement of $d/t$ is 
set very high, then the decision to measure $d/t$ will never be taken. In this example, it is not possible to reach 
the stopping criterion given in \eqnref{eq:ex_4_stopping} without at least one such measurement, and hence, 
the myopic strategy will keep requesting measurements of $X_{\text{m}}$ and evaluations of $p_{\text{FE}}$
indefinitely, accumulating a potentially infinite cost. This is indeed a drawback of the myopic strategy, 
which could be alleviated by looking multiple steps ahead, and at least through a full dynamic programming 
implementation. 

\section{Concluding remarks}
\label{sec:concluding_remarks}
We have presented a general formulation of the Bayesian optimal experimental design problem for structural reliability 
analysis, based on separation of the aleatory uncertainty or randomness associated with a given structure, 
and the epistemic uncertainty that we wish to reduce through experimentation. 
The effectiveness of a design strategy is evaluated through a \emph{measure of residual uncertainty}, 
and efficient approximation of this quantity is crucial if we want to apply algorithms 
that search for an optimal strategy. 
Our proposed approach makes us of a pruned importance sampling scheme for subsequent estimation of 
(typically small) failure probabilities for a given epistemic realization, combined with 
the unscented transform epistemic uncertainty propagation.
In our numerical experiments, we made use of a rather naive implementation of the unscented transform, 
in the sense that the number of sigma-points is very low, and that these are determined a priori with 
a deterministic procedure. Since the alternative by \cite{Merwe:2003:sigmapoint} produced satisfactory results in all of  
our numerical examples, no further consideration was made with respect to alternative methods for 
sigma-point selection. From applications to Kalman filtering, it has been observed that this version of the 
unscented transform has a tendency to over-estimate the variance, which is something we notice also in our experiments. 

For the application we consider in this paper, we emphasize that the unscented transform is used as a \emph{proxy} for the measure 
of residual uncertainty to be used in optimization, as a numerically efficient alternative 
that should be proportional to the true objective. Hence, we view the unscented transform as a tool to
find the best decision or strategy, where we get the possibility of exploring 
many decisions approximately rather than a few exactly.
Once an optimal strategy is found, we estimate 
the corresponding measure of residual uncertainty using a pure Monte Carlo alternative which is exact in the limit.
We note that for global optimization of acquisition functions, we have used a combination of random sampling and gradient based local optimization. With this procedure, an optimization objective given by $H_{3, k}$ (and also $H_{2, k}$) is generally more suitable than $H_{1, k}$, as it is less susceptible to noise coming from Monte Carlo estimation (see for instance \figref{fig:example_2_1}). 
On the other hand, $H_{1, k}$ has a natural interpretation (the variance of the failure probability), and is therefore 
a better measure for evaluating convergence, or for early stopping as discussed in \secref{sec:example_4}.

Although we focus on the estimation of \emph{a failure probability} in this paper, 
many of the main ideas should also be applicable 
for other estimation objectives using models where a hierarchical structure can be utilized. 
For instance, when $\alpha_k$ is some other quantity of interest depending on the random variable $g(\X)$, 
not necessarily given by an indicator function as in \eqnref{eq:pof_def}. 
For the applications considered in this paper, we have assumed that an isoprobabilistic transformation of $\X$ 
to a standard normal variable is available, which is often the case in structural reliability models. 
We make use of this assumption only to apply some well known techniques for failure probability estimation, 
but note that other alternatives, for instance the one presented in \appref{sec:app:crude_sampling}, 
can be used instead. 

There are several ways to improve the methodology presented in this paper. 
For instance, other alternatives of the unscented transform could be applied, see for instance \cite{Menegaz:2015:Systematization_UT}, 
or the parameters determining the set of sigma-points used in this paper could be optimized 
as in \citep{Turner:2010:learn_sigma_points}.

As seen in Section \ref{sec:example_4}, the myopic, one-step look ahead strategy, can make it impossible to reach the stopping criterion of the algorithm. As mentioned, a way to avoid this problem is by looking at the whole dynamic programming formulation \eqref{eq:dyn_prog}. However, this formulation suffers from the curse of dimensionality. Since the myopic formulation corresponds to truncating the sum in the dynamic programming formulation \eqref{eq:dyn_prog} to only one term, it is of interest to study methods where more terms of the sum are included (multi-step look ahead). How much better do the estimations get by including an extra term, and how much does the computation time increase? Is it possible to determine an optimal choice of truncation where we weigh accuracy and computation time against one another? 
Different ways of finding approximate solutions to the complete dynamic programming problem has been the focus of much research 
within areas such as operations research, optimal control and reinforcement learning, 
and trying out some of these alternatives is certainly interesting avenue for further research.

Another interesting topic worth investigating is how the numerical examples in this paper compare to the case where we estimate the \emph{buffered failure probability} instead of the classical failure probability. Buffered failure probabilities were introduced by Rockfellar and Royset \cite{RockafellarRoyset} as an alternative to classical failure probabilities in order to take into account the tail distribution of the performance function. See Dahl and Huseby \cite{DahlHuseby} for an application of this concept to structural reliability analysis.

One may also discuss whether using heuristic optimization objectives chosen to approximate the variance is reasonable. By essentially focusing on minimizing the variance of the failure probability, we say that all deviations from the true value is equally bad. In reality, overestimating the failure probability can be costly, but is not nearly as problematic as underestimating the failure probability. Because of this, the variance may not be the most appropriate measure of risk. It would be interesting to also derive heuristic optimization objectives based on approximating other risk measures.

These questions are of interest, but beyond the scope of the current paper, and the topics are left for future research.

\begin{acknowledgements}
This work has been supported by grant 276282 from the Norwegian Research Council and 
DNV GL Group Technology and Research (Christian Agrell),
and by project 29989 from the Research Council of Norway as part of the SCROLLER project (Kristina Rognlien Dahl).
\end{acknowledgements}

\bibliographystyle{spbasic_nourl}
\bibliography{references}

\newpage
\begin{appendices}

\section{Gaussian process surrogate models}
\label{app:GP}
Here we briefly review the Gaussian process (GP) surrogate model in its canonical form, 
for Bayesian nonparametric function estimation.
For a broader overview of the relevant theory see e.g. \cite{Rasmussen:2005:GPML}.
For applications related to uncertainty quantification (UQ) dealing with deterministic computer simulations, \cite{Kennedy:2001:Cal} is a classical reference.

Let $f : \mathbb{X} \rightarrow \mathbb{R}$ denote a function that we want to estimate, 
and assume that a set of $k$ observations $(\x_1, \y_1), \dots, (\x_k, \y_k)$ have been made. 
For instance, evaluating $f(\x)$ could correspond to running a deterministic (and time consuming) 
computer simulation, in which case noiseless observations, $\y_i = f(\x_i)$, can be obtained. 
Alternatively, $f(\x_i)$ could correspond to some physical experiment, resulting in a noise perturbed 
observation $\y_i$. A GP surrogate model $\xi$ of $f$ is a tool to make inference 
about the value of $f(\x^*)$ for any new input $\x^* \in \mathbb{X}$, conditioned on the 
set of observations $(\x_1, \y_1), \dots, (\x_k, \y_k)$.

A Gaussian process $\xi$ indexed by some set $\mathbb{X}$ is defined by the property that 
for any finite subset $\{ \x_1, \dots, \x_N \}$ of $\mathbb{X}$, $\left(\xi(\x_1), \dots \xi(\x_N)\right)$ is an $N$-dimensional Gaussian random variable. 
We will view $\xi$ as a Gaussian distribution over real-valued functions defined on $\mathbb{X}$ 
(such as $f(\x)$).
Here $\mathbb{X}$ can be arbitrary but typically $\mathbb{X}$ is a subset of $\mathbb{R}^n$.
The GP $\xi$ is uniquely defined by its mean function $\mu(\x) = E[\xi(\x)]$ and 
covariance function $c(\x, \x') = E[(\xi(\x) - \mu(\x))(\xi(\x') - \mu(\x'))]$.
Hence, any function $\mu : \mathbb{X} \rightarrow \mathbb{R}$ paired with a 
positive semidefinite function $c:\mathbb{X} \times \mathbb{X} \rightarrow \mathbb{R}$ defines a GP, 
which we will denote $\xi \sim \mathcal{GP}(\mu, c)$.

Let $X = (\x_1, \dots, x_k)$, $Y = (\y_1, \dots, \y_k)$ denote the observations and
assume that $\y_i$ comes with additive Gaussian noise, $\y_i = f(\x_i) + \epsilon_i$ 
where $\epsilon_i$ are i.i.d. zero-mean Gaussian with common variance $\sigma^2$.
In this scenario, the conditional process $\xi | X, Y$ is still a Gaussian process. 
In particular, if $X^* = (\x^*_1, \dots, x^*_m)$ contains $m$ new input locations in $\mathbb{X}$, 
then the distribution of $\bm{\xi}^* = \xi(X^*) = (\xi(\x^*_1), \dots, \xi(\x^*_m))$ 
given the observations $X, Y$
is Gaussian with the following mean 
\begin{equation}
    \label{eq:GP_def_1}
    \begin{split}
        &E[\bm{\xi}^* | X, Y] = \mu(X^*)\\ &+ c(X^*, X)[c(X, X) + \sigma^2 I_m]^{-1} (Y - \mu(X)), 
        \end{split}
\end{equation}
and covariance
\begin{equation}
    \label{eq:GP_def_2}
    \begin{split}
        &\text{Cov}(\bm{\xi}^* | X, Y) = c(X^*, X^*)\\ &- c(X^*, X)[c(X, X) + \sigma^2 I_m]^{-1} c(X^*, X)^T.
    \end{split}
\end{equation}
Here $\mu(X^*)$ and $\mu(X)$ are vectors with elements $\mu(\x^*_i)$ and $\mu(\x_i)$ respectively, 
$I_m$ is the $m \times m$ identity matrix, and $c(X^*, X^*)$, $c(X^*, X)$ and $c(X, X)$ 
have elements $c(X^*, X^*)_{i, j} = c(\x^*_i, \x^*_j)$, $c(X^*, X)_{i, j} = c(\x^*_i, \x_j)$ 
and $c(X, X)_{i, j} = c(\x_i, \x_j)$.

For the scenario where observations are noiseless, $\y_i = f(\x_i)$, the distribution 
of $\bm{\xi}^* | X, Y$ is obtained with $\sigma = 0$ in \eqref{eq:GP_def_1}-\eqref{eq:GP_def_2}.

To define a GP prior $\xi \sim \mathcal{GP}(\mu, c)$ over functions $f : \mathbb{X} \rightarrow \mathbb{R}$, we need to specify the mean and covariance function. 
These are generally given as $\mu(\x | \theta)$ and $c(\x, \x' | \theta)$, conditioned on some parameter $\theta$. An appropriate value for $\theta$ is usually found through maximum likelihood 
estimation or cross validation using the set of observations $X, Y$. 
A fully Bayesian approach could also be pursued, where the posterior calculations typically 
involve Markov chain Monte Carlo as the formulation in \eqref{eq:GP_def_1}-\eqref{eq:GP_def_2} is not sufficient. 
In the numerical experiments presented in this paper, we have made use of a constant mean function 
and a Mat\'ern $5/2$ covariance function using plug-in hyperparameters 
$\theta = (\sigma_c, l_1, \dots, l_n)$
determined from maximum likelihood estimation.
The Mat\'ern $5/2$ covariance function for $\x, \x' \in \mathbb{R}^n$ is defined as 
\begin{equation}
\begin{split}
    &c(\x, \x') = \sigma_c^2 (1 + \sqrt{5}r + \frac{5}{3}r^2) e^{-\sqrt{5}r}, \\ 
    &r = \sqrt{\sum_{i = 1}^n \left( \frac{x_i - x'_i}{l_i} \right)^2}.
\end{split}
\end{equation}

\section{The sampling distribution $q_{\X}$}
\label{app:sampling_dist_q}
Here we present some further details on how the set of samples $\{ \x_i, w_i \}$ in \secref{sec:sampling} can be generated. We start by reviewing some classical techniques from 
structural reliability analysis that are based on finding 'important' regions in $\mathbb{X}$. 
The sampling distribution $q_{\X}$ used in this paper is then defined in \secref{sec:app:DP_sampling}.
It is based on the assumption that $\X$ can be transformed to a standard multivariate Gaussian variable $\bm{U}$, and that $q_{\bm{U}}$ can be constructed by solving a set of constrained optimization problems in $\bm{U}$-space. For the scenario where these assumptions do not hold, we present 
an alternative approach in \secref{sec:app:crude_sampling}, which is based on a 
naive exploration of the $\mathbb{X}$-space. Although this will require evaluation of a larger 
set of samples of $\X$, no optimization is required and numerical implementation is straightforward. 

\subsection{Local approximations in SRA}
\label{sec:app:local_SRA}
In \secref{sec:approx_H} we briefly discussed the challenges with estimation of the failure 
probability $\bar{\alpha}(g)$ in \eqnref{eq:pof_def}. A different alternative often used in structural reliability analysis, is to 
approximate the performance function $g(\x)$ with a function $\hat{g}$ where $\bar{\alpha}(\hat{g})$ can be computed analytically. 
In this scenario, it is convenient to transform $\X$ to a standard normal variable $\bm{U}$. 
We will let 
\begin{equation}
    \label{eq:isoprob}
    \X \xrightarrow[]{\mathcal{T}} \bm{U} \sim N(\bm{0}, I)
\end{equation}
denote an isoprobabilistic transformation, 
where $\bm{U} = \mathcal{T}(\X)$ is multivariate standard Gaussian with $\text{dim}(\bm{U}) = \text{dim}(\X)$. 
Note that for any univariate random variable $X$ with CDF $F(X)$, a transformation of this
type available as $\mathcal{T}(X) = \Phi^{-1}(F(X))$. 
The generalization to multivariate $\X$ is the Rosenblatt transformation, where 
$\bm{U}_i = \Phi^{-1}(F_i(\X_i | \X_1, \dots, \X_{i-1}))$. 
In structural reliability problems, it is often natural to define $\X$ in terms of 
the marginal distributions and a copula, in which case the isoprobabilistic transformation 
\eqnref{eq:isoprob} can be simplified. A common alternative is to use a Gaussian copula, 
where \eqnref{eq:isoprob} can be obtained using the Nataf transformation \citep{Lebrun:2009:Nataf}. 

In the following we let $g(\uu)$ denote the function $g(\cdot)$ applied to $\x = \mathcal{T}^{-1} (\uu)$.
Methods such as FORM (First Order Reliability Method) and SORM (Second Order Reliability Method) 
make use of local approximations in the form of a linear or quadratic surface fitted to 
$g(\uu^*)$ at a certain point $\uu^* \in \mathbb{R}^n$. This point $\uu^*$ is often called the 
\emph{design point} or \emph{most probable point} (MPP), and it is defined as 
\begin{equation}
    \label{eq:MPP_opt}
    \uu^* = \argmin_{\uu \in \mathbb{R}^n} \{\norm{\uu} \ | \ g(\uu) \leq 0 \}.
\end{equation}
Observe that if $\hat{g}(\uu)$ is the first-order Taylor approximation of $g(\uu)$ at $\uu^* $, 
i.e. $\hat{g}(\uu) = g(\uu^*) + \nabla_{\uu}g(\uu^*)(\uu - \uu^*)$, then 
$\bar{\alpha}(\hat{g}) = \Phi(-\norm{\uu^*})$, and this is an upper bound on the failure 
probability if the failure set is convex in $\bm{U}$-space.

In \secref{sec:sampling} we discussed the importance sampling estimate of the failure probability
given some proposal distribution $q$. A natural candidate is to let $q$ be a distribution 
centered around the design point, $\uu^*$ in $\bm{U}$-space or $\x^* = \mathcal{T}(\uu^*)$
in $\X$-space. The alternative where the estimation is performed in $\bm{U}$-space 
with $q_{\bm{U}}(\uu) = \phi(\uu + \uu^*)$ is often used in practice. 
For a more detailed discussion around this kind of sampling, the local approximations 
and structural reliability analysis in more general, see for instance \cite{madsen2006methods} or \cite{Huang:2017:SRA}.

The constrained optimization problem \eqnref{eq:MPP_opt} plays an important role in 
structural reliability analysis. Although any general-purpose algorithm can be used,
customized algorithms that take advantage of the special form of the objective function 
are recommended. Various alternatives have been developed for this purpose, 
see for instance \cite{Gong:2011:robust_iter_alg} and the references therein. 
For the applications in this paper we have made use of the 
iHL-RF method from \cite{Zhang:1995:iHLRF}.

\subsection{The design point mixture}
\label{sec:app:DP_sampling}
We observe first that a solution to \eqnref{eq:MPP_opt} is 
not necessarily unique, and also that multiple \emph{local} minima may exists when the performance 
function is nonlinear. Most algorithms designed to solve \eqnref{eq:MPP_opt} numerically 
start with some initial guess $\uu_0$, and take iterative steps until a minimum is obtained. 
To reduce the risk of overestimating $\norm{\uu^*}$, multiple restarts with different 
(possibly randomized) initial guesses $\uu_0$ is often applied. 

Given a finite-dimensional approximation of a performance function $\hat{\xi}(\x, \E)$, 
we want to find a proposal distribution $q$ that is appropriate for a range of different
realizations $\e$ of $\E$. In particular, if $\{ (v_{j}, \e_{j}) \ | \ j = 1, \dots, M \}$
is the set of sigma-points for $\E$ as introduced in \secref{sec:UT_epistemic}, 
we want a set of samples from $q$ to be applicable for estimation of $\alpha(\hat{\xi}(\x, \e_j))$
for any $1 \leq j \leq M$.

For any $\e_j$, we will let $\uu^*_{1, j}, \dots, \uu^*_{N, j}$ denote $N$ design points in 
$\bm{U}$-space corresponding to $\hat{\xi}(\x, \e_j)$, obtained using randomized initialization. 
(Note that for methods such as iHL-RF, it is also reasonable to use $\uu^*_{i, j}$ as an 
initial guess in the search for $\uu^*_{i, j+1}$). 
We then define $\bm{Q}$ as the equal-weighted Gaussian mixture of the $NM$ random variables 
$\bm{Q}_{i, j} = \bm{U}_{i, j} + \uu^*_{i, j}$, where $\bm{U}_{i, j}$ are i.i.d. 
standard multivariate Gaussian. Sampling from $\bm{Q}$ is then straightforward, 
and importance sampling estimates can be obtained in the $\bm{U}$-space 
using $p_{\bm{U}}(\uu) = \phi(\uu)$ and $q_{\bm{U}}(\uu) = \frac{1}{NM} \sum_{i, j} \phi(\uu - 
\uu^*_{i, j})$, where $\phi$ is the multivariate standard normal density.

\subsection{A simple alternative}
\label{sec:app:crude_sampling}
The sampling strategy presented in \secref{sec:sampling} is based on 
1) generating a set of samples that should "cover relevant locations" in the input space $\mathbb{X}$, and
2) \emph{prune} the set of samples using a threshold on the measure of insignificance \eqnref{eq:eta_def}. 

The "relevant locations" in the first step is typically somewhere in the "tail" of the 
distribution of $\X$, where also the (uncertain) performance function $\hat{\xi}_k(\x)$ may 
be close to zero. In \secref{sec:app:DP_sampling} we made use of importance sampling 
around design points, which is a common technique in structural reliability analysis. 
As a simple alternative, we can let $q$ be any distribution from which it is 
easy to generate samples covering the effective support of $p_{\X}$ (i.e. a bounded domain 
where $\X$ lies with probability $\approx 1$).
For instance, assuming $\bm{U}$ is $n$-dimensional standard normal (e.g. $\bm{U} = \mathcal{T}(\X)$ 
if the isoprobabilistic transformation is still applicable), 
we could let $q$ be a uniform density on the hypercube $[-b, b]$
where $b = \Phi^{-1}(1-p_{min})$ for some absolute lower bound on the failure probability $p_{min}$.

Because the initial set of $N$ samples from $q$ will be reduced to a fixed number of $n$ samples 
after the pruning step, this is a viable alternative. 
However, in order to obtain similar importance sampling variances (see \eqnref{eq:split_MC_var}) 
as with the method in \secref{sec:app:DP_sampling}, the initial number 
of samples $N$ (and hence the number of evaluations of the pruning criterion $\eta(\x)$) will 
have to be larger. 

\section{Selecting sigma-points for the unscented transform}
\label{app:sigma_points}
Here we briefly review the method for sigma-point selection by \cite{phd:Merwe} and 
present the sigma-points used for the numerical experiments in \secref{sec:num_exp}.

According to \cite{Labbe:2014:Kalman_python}, research and industry have mostly settled on the 
version published in \citep{phd:Merwe}. 
Here, the sigma-points are given as a function of the mean and covariance matrix of 
the input variable, together with three real-valued parameters $\alpha, \beta$ and $\kappa$. 
In the case where $\bm{U}$ is a standardized $n$-dimensional random variable with $E[\bm{U}] = \bm{0}$ and $E[\bm{U}^2] = I$, 
we obtain $2n+1$ points $\uu_i$ are as follows
\begin{equation*}
    \begin{split}
        &\uu_0 = \bm{0}, \\ 
        &\uu_i = \alpha \sqrt{n + \kappa} \bm{\nu}_i, \\ 
        &\uu_{i+n} = -\uu_i, 
    \end{split}
\end{equation*}
for $i = 1, \dots, n$ where $\bm{\nu}_i = (0, \dots, 1, \dots, 0)$ is the standard unit vector in $\mathbb{R}^n$.
Two different sets of weights are used with this procedure, one for the mean and one 
for the covariance in \eqref{eq:UT_moments}. 
We denote these $v_i^m$ and $v_i^c$ respectively, and they are given as 
\begin{equation*}
    \begin{split}
        &v_0^m = 1 - \frac{n}{\alpha^2 (n + \kappa)}, \ \ v_0^c = v_0^m + 1 - \alpha^2 + \beta,\\
        &v_i^m = v_i^c = \frac{1}{2 \alpha^2 (n + \kappa)} \textit{  for } i = 1, \dots, 2n. 
    \end{split}
\end{equation*}

For Gaussian distributions, it is often recommended to set $\beta = 2$, $\kappa = 3-n$ and let $\alpha \in (0, 1]$.
In the numerical examples presented in this paper we have used this set of parameters with $\alpha = 0.9$.

\end{appendices}

\end{document}